\documentclass[11pt]{article}
\usepackage[letterpaper,margin=1in]{geometry}
\usepackage[T1]{fontenc}
\usepackage{lmodern,microtype,amsmath,amssymb,amsthm,booktabs,tabularx,array,enumitem}
\usepackage{xcolor,tikz}
\usepackage{xurl}
\usetikzlibrary{arrows.meta,positioning}
\usepackage[colorlinks=true,linkcolor=blue!45!black,urlcolor=blue!45!black,citecolor=blue!45!black]{hyperref}
\hypersetup{pdftitle={Matoi: Selective Revalidation of Consequence Evidence for Agent Authorization},pdfauthor={Zexun Wang},pdfsubject={Bounded runtime authorization, checked consequence evidence and recovery},pdfkeywords={agent authorization, evidence reuse, runtime governance, conservative recovery}}
\usepackage{fancyhdr}
\newcolumntype{Y}{>{\raggedright\arraybackslash}X}
\setlist{nosep,leftmargin=*}
\newtheorem{proposition}{Proposition}

\newcommand{\Tr}{\mathcal{T}}

\newcommand{\Valid}{\operatorname{Valid}}
\newcommand{\Scope}{\mathcal{B}}
\title{Matoi: Selective Revalidation of Consequence\\Evidence for Agent Authorization}
\author{Zexun Wang\\Ond Holdings Inc.\\Toronto, Ontario, Canada\\
\href{mailto:jw@nd.im}{\texttt{jw@nd.im}}\\
Business alternative: \href{mailto:jw@osuite.ai}{\texttt{jw@osuite.ai}}\\
ORCID: \href{https://orcid.org/0009-0009-9098-793X}{0009-0009-9098-793X}}
\date{15 October 2026\\\small Technical preprint --- not peer reviewed}
\begin{document}
\maketitle
\begin{abstract}
Matoi separates reusable consequence evidence from current execution authority
and durable occurrence history. Its finite query-aware contract covers branches,
ordered collections and negative lookups. A consumer reuses only its own checked
evidence whose observations remain valid; signed requests, live grants, quota
and uncertain outcomes remain separate obligations. Conditional correctness
arguments distinguish locally serialized claims from nontransactional remote
effects. Two studies strengthen cost attribution and domain evaluation.
A 1,440-cell matched-wire study repeats 144 authored inputs. With compact input
bytes identical, reuse reduces complete lifecycle cost by 3.49 percent in JS and
1.79 percent in Python. Combined representation/reuse reductions against full
proof are 10.87 and 16.53 percent; five of twelve family intervals for the matched
comparison cross one in each language. A prospectively frozen AgentDojo Banking
mapping covers 5/16 user plans and 7/9 attack plans. Both governed strategies
preserve five exact benign transfers and refuse 35 reused attack-plan/policy
crosses. Pre-satisfied utility predicates are reported separately from completed
work. These are original scripted synthetic tools, not autonomous agents or a
real bank. Earlier studies retain cheaper protected-cache alternatives,
policy-compatible phishing, provider duplication and post-claim revocation.
Separate JS/Python consumers establish finite-profile interoperability within
the same project; durable recovery rechecks evidence and retains uncertain
reservations. The results support a bounded, auditable reuse and execution
contract, not a new general incremental algorithm, arbitrary-tool closure,
atomic third-party commitment or production superiority.
\end{abstract}

\section{Introduction}
Approving a consequential agent action is not equivalent to approving its textual invocation. A support-resolution operation may issue a refund, retrieve a document, and invoke a notification workflow. Its outward name can remain unchanged while a configuration update adds a data export. An authorization decision therefore depends on more than the entry point. It also depends on which consequences the environment can activate and which of those consequences the authority actually granted.

Recent work makes this distinction concrete. Agent Approval Laundering studies the coverage relation between a decision-time approval object and an invocation's transitive execution closure, and proposes Bind--Enforce--Contain~\cite{zhang}. Formation-Consistent Dispatch (FCD) separates effect compatibility from formation-scoped interpreter authority and lifecycle eligibility~\cite{fcd}. These are complementary obligations. We do not claim that adding typed effects, an enforcement step, or a graph is by itself a contribution beyond them.

Matoi addresses repeated use in a changing environment. Whole-snapshot invalidation is easy to explain but can repeat analysis after unrelated changes. Reusing an old approval because the entry identity is unchanged can admit new consequences. Between these extremes, a runtime needs evidence that a previous analysis remains applicable, together with an explicit refusal path when that evidence is missing. Crucially, this is a question about reusing a computation, not a license to cache away revocation, shared quotas or the exact effect that will commit.

\paragraph{Research question.}
Under a declared effect-coverage model, when does consequence evidence remain applicable across runtime state changes, without widening requested authority or replacing current commitment checks, and when is reuse cheaper than reanalysis?

\paragraph{Contributions.}
\begin{enumerate}
\item An applicability contract separating a fixed authority envelope, prospective consequence evidence and current execution obligations. Conditional arguments distinguish receiver-serialized commitment from local dispatch to a nonparticipating remote provider, without equating applicable evidence with current permission.
\item Two explicitly bounded semantic realizations: a straight-line workflow implementation with a read-footprint argument and transactional receiver, and a query-aware finite kernel with checked summary ingestion. A durable extension of the latter reconstructs evidence after restart and preserves uncertain dispatch reservations under a single-history recovery contract.
\item An evidence-separated evaluation of original process execution, finite semantics, real GitHub outcomes and local recovery, followed by a four-strategy, same-kernel qualification. Full-cost, offered-load, retention and fault matrices retain strong-cache ties, refusals and unknowns. A supporting control-placement study tests the complete-mediation premise rather than assuming graph coverage implies deployed enforcement.
\item A separately versioned portable evidence profile, checked by JS and Python consumers, and an explicit durable integration. Their later cost and fault studies preserve the boundaries of the original experiments. A new three-arm matched control holds compact inputs identical while distinguishing complete rule rechecking from accepted-subtree reuse; it measures a modest residual benefit without replacing the historical transport comparison.
\item Two narrow externally authored AgentDojo plan studies, for Slack messages and Banking transfers. The new Banking study reports whole-plan eligibility, exact mappings, actual local mutations, refusal layers and original pre/post utility separately; repeated attack-policy crosses and pre-satisfied predicates are not independent attacks or recovered task utility.
\end{enumerate}
The central empirical lesson is not that dependency caching is always faster. It is that evidence applicability and authorization validity are different obligations with different costs. The earlier four-strategy same-kernel comparison does not establish a net cost advantage over a protected reference cache; the new matched compact-packet control finds a small benefit against full rechecking, not superiority over competent external implementations. A falsely complete deployment inventory defeats a coverage claim. Exact query tracking and validation interpreters are established techniques, not a claimed new algorithm. General dependency discovery, distributed authority and independently validated workflow generalization remain outside the demonstrated result. Successor studies do not retroactively strengthen the semantics or trust boundary of the original measurements.

\paragraph{Reading the paper.}
The refund example explains why an unchanged action may still need a different decision. The formal model then separates what may happen, what is permitted and what was observed. The original studies precede the query-aware kernel, GitHub pilot and durable recovery. Section~\ref{sec:qualification} gives the earlier matched kernel comparison; Section~\ref{sec:qualified-coverage} tests the coverage premise with real queued mutations. Section~\ref{sec:portable-update} adds the later portable evidence and integration studies rather than attributing their results to the original prototypes. Section~\ref{sec:agentdojo} retains the Slack plan study. Sections~\ref{sec:matched-cost} and~\ref{sec:banking} then add the separately frozen matched-cost control and Banking results before the external-validity discussion. Readers interested in deployment should attend especially to commitment ownership, refusal and unknown outcomes; readers interested in optimization should distinguish preparation cost from total workflow time.

\section{Relationship to Existing Work}
PCAA concerns deployer-owned authority and proof-carrying actions~\cite{pcaa}; CAVA concerns canonical action representation across runtime syntax~\cite{cava}. The current OSuite reference path adds exact execution-context commitments, coupled operation/effect/resource/destination constraints, prepared admission, and occurrence-level evidence. BAF and AREG provide useful implementation and research components, but their combination is not itself evidence of novelty.

Zhang et al. explicitly bind source and state evidence before authorization, enforce policy, and contain permitted execution. Their six-class effect projection is an evaluation vocabulary, not an intrinsic limit: the paper explicitly permits destination, argument, credential, severity, timing and outcome refinements and discusses SaaS and nested-agent settings~\cite[Sec.~VII]{zhang}. Our comparison must therefore be against this complete framing, not ``coarse measurement without enforcement.''

FCD already models rich effects, configuration and dependency identities, compatibility evidence, formation-time target authority, and a final-hop fence~\cite[Sec.~IV]{fcd}. A compatible future implementation is not automatically authorized for a previously formed call. Matoi fixes the executable target in its initial fragment and treats reanalysis of permitted state variation separately from authority to replace the interpreter. An extension must preserve FCD's formation and lifecycle obligations rather than bypass them using a broad effect envelope.

FCD also reports reusable reviewed profiles over official tool releases,
a frozen release expansion checked by a separate source oracle, verified-capability
caching, and durable lease markers for uncertain outcomes. These are relevant
comparators, not missing features that Matoi introduces. Our evaluated distinction
is narrower: checked query dependencies under state changes within a fixed finite
interpreter, composed with an explicit recovery contract. We have not established
an advantage over FCD's implementation or its reviewed profile pipeline.

Recent agent-specific work further limits a novelty claim based on execution
gates alone. AgentSpec specifies triggers, predicates and enforcement actions
for runtime constraints~\cite{agentspec}. VEPA studies the verified evolution
and evidence-conditioned activation of formal policies from deployment
feedback~\cite{vepa}. APEX compiles a trusted-task authorization contract and
checks both proposed effects and the endorsed use of runtime information at
execution boundaries~\cite{apex}. Policy learning, information endorsement and
prospective consequence-evidence reuse are distinct questions. Our finite
publication profile does not establish the first two or an attack-success
advantage over these systems. Its comparison needs separately fixed obligations
and workloads, not their reported scores placed beside our conformance counts.

Operational systems also offer more than logs and tool-name permissions.
AgentCore Policy documents deterministic input rules and session-scoped temporal
conditions for approval correlation, action counts and cumulative values;
policy changes can invalidate temporal sessions~\cite{agentcore-temporal}.
Anthropic describes environmental containment, credential separation and the
limits of per-action human approval~\cite{claude-containment}. These are
deployment reference points, not benchmarked baselines in this paper. Neither
external enforcement nor moving work ahead of an action is independently novel.
Any Matoi composition must retain the downstream enforcement and identity
conditions of its actual deployment.

HarnessSecurity-Bench evaluates attack effects separately from legitimate task
utility across coding harnesses and configuration settings~\cite{harness-security}.
That distinction matters here: a blocked consequence can preserve safety while
preventing useful work, and a local oracle match is not evidence of semantic
generalization. The qualification below retains both outcomes and complete
evidence-preparation, dispatch and reconciliation costs; independently defined
real-workflow evaluation remains a separate requirement.

Capability attenuation, proof-carrying code, abstract interpretation, incremental dependency analysis and complete mediation provide relevant foundations~\cite{saltzer,cousot,necula,macaroons,acar}. Their established ideas make a graph-bound token insufficient as a novelty argument. Matoi's propositions are conditional composition arguments, not replacements for those foundations. In particular, authenticating evidence does not establish that its analyzer is sound.

Checked ingestion narrows a trust boundary within the finite implementation; it does not make conventional dependency caches unsound. A trusted analyzer with protected storage is a legitimate alternative, and guarded summaries, validation interpreters and incremental computation remain strong comparators. Reconstructing an exact query footprint supplies no claim to a new algorithm, more permissive semantic equivalence, or completeness for arbitrary programs. Those distinctions require separately fixed workloads and matched implementations, not additional authored repetitions.

Here a guarded-summary baseline means a protected analyzer result with explicit
applicability predicates over its effect-determining inputs. Checking a candidate
result against a reference computation is related to translation validation,
which validates individual translations rather than relying only on a trusted
optimizing translator~\cite{translation-validation}. Our checker compares finite
query observations and effect traces; it does not validate arbitrary compiler
output or establish binary equivalence.

Durable intent and recovery build on established transaction and messaging
ideas. ARIES is a foundational write-ahead recovery method~\cite{aries}, not an
algorithm we claim SQLite implements. Helland distinguishes local transactions
from messaging across independently managed entities~\cite{helland}; the
transactional-outbox pattern durably couples a local change with a message
awaiting publication~\cite{richardson}. Matoi's uncertain-send rule is deliberately
more conservative than retrying an outbox message: without an adequate provider
idempotency/observation contract, it retains unknown rather than sending again.
These mechanisms are foundations and necessary comparators, not implementations
experimentally outperformed by the present work.

Incremental abstraction-carrying code (ACC) is particularly close to the
later repair implementation~\cite{iacc}. Its consumer retains validated
answers and dependency arcs, checks untrusted program/certificate updates,
and propagates affected changes. Generic authenticated data structures also
support queries, updates and cached proof components~\cite{gpads}.
Incremental checking, compact proofs and consumer-owned validated state are
therefore not new individually or merely in combination. Matoi's narrower
subject is an exact finite consequence/query contract under data-snapshot
changes and its composition with live execution obligations. Complete trusted
snapshots remain an assumption, not authenticated remote state discovery.

AgentSpec's predicates include state and trajectory context; a competent
implementation can be extended with fresh consequence evaluation. Its stated
limitation in prospective multi-step analysis motivates that strengthened
comparison, not a claim that it is only a tool-name classifier. AgentDojo's
stateful task/security predicates illustrate another useful design principle:
measure legitimate task completion separately from attacker success~\cite{agentdojo}.
Its authored synthetic environments are not production traces, and externally
authored predicates are not external validation of a system run by us.

Usage control already treats ongoing authorization, mutable attributes, conditions and obligations as first-class concerns~\cite{ucon}. Matoi does not claim to introduce continuous authorization or revocation. Its narrower subject is the reuse of prospective consequence evidence while these mutable commitment obligations remain active. The finite implementation also falls short of general usage-control obligations: it enforces a fixed sequence and scalar budget, not arbitrary temporal duties or distributed policy administration.

ZeroGate moves exact-action preparation ahead of dispatch while retaining a durable admission boundary~\cite{zerogate}. That changes where authorization cost is paid, not whether all of it exists. Matoi studies a different reuse question: whether consequence evidence remains applicable after a state change. Combining the two requires both obligations. A prepared pass cannot authorize a changed effect, and applicable effect evidence cannot revive a revoked pass.

\subsection{The comparison unit is an obligation, not a brand}
The proposed contribution is not that other systems lack governance. A useful comparison asks which evidence justifies a particular decision and when that evidence stops applying. Table~\ref{tab:prior-obligations} identifies obligations that survive composition. The rows describe research relationships, not benchmarked implementations or claims that a cited work cannot be extended.
\begin{table}[ht]
\centering\small
\begin{tabularx}{\linewidth}{@{}>{\raggedright\arraybackslash}p{0.21\linewidth}YY@{}}
\toprule
Prior framing & Obligation retained & Question examined here \\
\midrule
Approval laundering & Represent and contain transitive consequences & When may an existing bounded consequence analysis remain applicable after a state change? \\
Formation-consistent dispatch & Effect compatibility does not grant interpreter authority & Can evidence refresh preserve the original authorized target and delegation? \\
Capability attenuation & Caveats restrict rather than manufacture authority & How are a fixed grant and changing evidence kept distinct at execution? \\
Incremental computation & Reuse needs dependency correctness and a cost model & Which reusable computation is evidence, and which live checks cannot be replaced by that computation? \\
Reference monitors & Every covered effect must cross the protection boundary & Which mutation and authority update share the declared serialization point? \\
\bottomrule
\end{tabularx}
\caption{Claim-by-claim comparison. Typed effects, attenuation, dependency tracking and final enforcement are established ideas. The investigated contribution is their explicit applicability contract and measured trade-off, not a claim of inventing those mechanisms.}
\label{tab:prior-obligations}
\end{table}

\begin{table}[ht]
\centering\small
\begin{tabularx}{\linewidth}{@{}>{\raggedright\arraybackslash}p{0.26\linewidth}Y@{}}
\toprule
Evidence design & Trust and deployment obligation \\
\midrule
Trusted full evaluation & Trust a complete local evaluator for each snapshot;
no untrusted producer certificate is required. \\
Protected caching / incrementality & Keep complete dependency tracking,
retained computations and invalidation state inside the trusted boundary. \\
Compact complete rechecking & Accept untrusted references only into the
consumer's checked root, then rederive every resolved rule. \\
Untrusted-producer checked reuse & Protect the same custody; check replacement
rules and current applicability of retained subtrees before accepting the root. \\
\bottomrule
\end{tabularx}
\caption{Trust/baseline comparison. All designs still need live authority and
occurrence checks. Measured comparators are same-project implementations and
internal ablations; no measured external-baseline advantage is established.}
\label{tab:trust-baselines}
\end{table}

\section{Worked Example and Threat Model}
\subsection{One invocation, several different reasons to reconsider it}
Consider a support workflow that refunds 7,300 US-dollar cents and queues a notification to an internal helpdesk. The skill requests a per-refund ceiling of 10,000 cents. The deployer grants that request only for a fixed interpreter, supported workflow language, destination set and declared state variation. Current runtime policy limits the refund to 8,000 cents. A shared ledger separately limits cumulative refunds; the per-operation ceiling is not that shared budget.

The action name remains \texttt{resolve\_case}. The prospective analysis reads the refund amount and the two destinations, including the queued child's destination. It does not read a dashboard theme. The notification's permitted text is bounded in size but not classified for confidentiality in this fragment: an allowed internal destination is not a general data-loss-prevention guarantee.

\begin{table}[ht]
\centering\small
\begin{tabularx}{\linewidth}{@{}>{\raggedright\arraybackslash}p{0.24\linewidth}YY@{}}
\toprule
Change after preparation & Evidence consequence & Authorization consequence \\
\midrule
Dashboard theme changes & Supported dependency projection is unchanged & Bound can be reused; fresh authority and budget checks remain \\
Refund becomes 5,000 cents & Relevant operand changes; recompute the effect & May remain within an explicitly delegated grant; a new occurrence still consumes its own budget \\
Notification points outside & Relevant destination changes & New analysis cannot expand the permitted destination; refuse the uncovered suffix or whole preparation \\
Runtime ceiling becomes 4,000 & Effect analysis can remain identical & The 7,300-cent refund is no longer authorized \\
Interpreter is replaced & Effects might remain identical & Compatibility is insufficient unless the original formation authority permits the replacement \\
Refund ACK is lost & The caller lacks an outcome, not necessarily an effect & Query the committed occurrence; do not infer permission to refund again \\
\bottomrule
\end{tabularx}
\caption{Different invalidation domains in one running example. These are semantic cases, not additional independent experimental observations.}
\label{tab:running-example}
\end{table}

This example separates an unchanged request, an applicable analysis, a valid grant and a still-permitted next effect. They are four different propositions. A safe gate can preserve the first two and nevertheless refuse because the fourth has changed. Conversely, rejecting every theme change confuses an environmental difference with a relevant change.

\subsection{Assets, attacker capabilities and trusted components}
The protected assets are the receiver's refund budget, exact destination constraints and occurrence integrity. The reference threat includes actor requests with altered effects or signatures, duplicate delivery, stale work, and changes to state or authority between preparation and a later effect. The study injects these through an authored schedule; it does not deploy an autonomous adversary or independently sampled attacks.

In the original transactional experiment, the issuer, analyzer implementation, controller-to-owner administration channel, selected interpreter, effect owner, database and operating-system isolation are trusted. The actor implementation is trusted to implement the closed DSL for the prospective-coverage argument. The receiver checks still reject the explicitly injected malformed or out-of-bound actor requests; this is not evidence of protection against an actor with host access to the owner's database or signing key. The later checked-ingestion kernel treats candidate summaries as untrusted data but trusts its reconstruction, shared parser/type checks, immutable snapshots and private storage. Its GitHub composition additionally trusts the credential-bearing dispatcher and provider readback; child environment filtering is not isolation from a malicious process running as the same OS user.

Authority changes are supplied by the trusted administration channel. The experiment does not authenticate a real enterprise administrator or prove that a customer originally intended the signed policy. A signature authenticates a recorded issuer, not the correctness of its judgment. The statement about all effects quantifies only over the configured mediated effect domain. Network traffic to an uninstrumented service, arbitrary binary behavior and side channels are outside that domain.

\subsection{Three distinct success conditions}
\textbf{Prospective coverage} means the supported workflow's intended trace is represented before enforcement truncates it. \textbf{Receiver safety} means every committed covered effect satisfies its current obligations. \textbf{Task success} means the requested workflow completes with an evidenced outcome. A missing notification can preserve safety while failing the task; a blocked hidden export can preserve safety while falsifying the earlier coverage claim. Reporting one of these as all three would make the evaluation circular.

\section{Objects, Authority and Evidence}
\subsection{Semantic identity is not execution identity}
Let $A$ identify the canonical entry action and let $o$ identify a concrete execution occurrence. Repeated executions may share $A$ but must not share an occurrence identity, a consumed nonce, or an unresolved outcome. Let $I$ identify an authorized executable target and $X_t$ the relevant runtime state at time $t$. Neither equality of $A$ nor equality of $I$ entails equality of all activated effects.

\subsection{Authority is a constraint on traces}
Let $S$ denote requested authority, $R_t$ the runtime's current authority, and $L$ the original leased envelope. A vendor-neutral skill contract is one possible source of $S$; this paper does not define or require the collaboration's evolving schema. For the safety argument, the fixed requested and leased languages are nonempty and prefix-closed: an allowed committed trace has allowed prefixes. This does not express liveness requirements such as eventually sending a receipt. Each boundary is a predicate on covered effect traces, with a remaining-budget interpretation when effects have already committed. For a new invocation, the effective boundary is
\begin{equation}
 \Scope_t = S \cap R_t \cap L. \label{eq:authority}
\end{equation}
The runtime may narrow or deny a request. It cannot use a prediction to expand $S$ or an old lease to override $R_t$. A new request or grant is a separate authorization event, not silent reuse.

Per-effect tuples are a useful restricted language but do not describe every trace constraint. Two refunds can each satisfy a \$100 per-operation limit while violating a \$100 aggregate limit. Ordering, delegation depth, multiplicity and cumulative data movement may also matter. Accordingly, the formal object is a trace boundary; the executable model below implements only finite tuples and a conservative scalar aggregate budget.

More precisely, after committed history $h$, the residual of a fixed trace language $Q$ is $Q/h=\{\sigma\mid h\sigma\in Q\}$. Let $R_t[h]$ denote the runtime's permission for future continuations given that history. The remaining boundary is
\begin{equation}
 \Scope_t[h]=(S/h)\cap(L/h)\cap R_t[h].
 \label{eq:residual}
\end{equation}
$R_t[h]$ is not automatically the quotient of a newly written full-history policy. A later runtime revision does not retroactively rewrite previously committed history. A shared budget is evaluated over the relevant joint history, not separately per worker; a continuation's local effects must be checked against that shared residual. The protocol must state which occurrences share a budget and its serialization order.

\subsection{Prospective bounds and retrospective evidence}
Let $\Tr(A,I,X_t;H)$ be the possible covered traces within a declared horizon/domain $H$. A prospective certificate provides a bound $U_t$ with the claimed coverage property
\begin{equation}
 \Tr(A,I,X_t;H) \subseteq U_t. \label{eq:coverage}
\end{equation}
The certificate must identify the analyzer, its supported language, dependency inputs, executable/receiver identities, validity conditions and exclusions. Authentication identifies the issuer; it does not prove Equation~\ref{eq:coverage}.

Observed lineage $O_t$ is a different object. It contains attempted, admitted, confirmed, denied and unknown occurrences with explicit parent/evidence relationships. Existing AREG research bounds reachability over a supplied finite set of occurrences and candidate lineage worlds. Such a bound can expose ambiguity without inventing edges. It cannot, by itself, establish all future behaviors of a program or discover uninstrumented SaaS actions. A forward analyzer or an independently justified containment contract is needed to connect observed relationships to $U_t$.

If multiple interpretations of evidence remain admissible, let $\mathcal{W}_t$ denote those interpretations. A safe automatic decision must establish the boundary for every admissible interpretation. Mixed or unresolved interpretations require review unless a sound conservative upper bound suffices. A missing record is not an empty consequence set. Correlated choices and aggregate budgets require trace-aware bounds; unioning graph edges can overapproximate impossible combinations and cause conservative refusal.

\begin{figure}[ht]
\centering
\begin{tikzpicture}[node distance=7mm and 8mm,box/.style={draw=blue!35!black,rounded corners=3pt,fill=blue!4,align=center,text width=0.40\linewidth,minimum height=13mm,font=\small},arr/.style={-{Latex},thick,draw=blue!40!black}]
\node[box] (skill) {Requested authority $S$\\Requested scope or skill contract};
\node[box,right=of skill] (evidence) {Prospective evidence $U_t$\\Dependencies, coverage, exclusions};
\node[box,below=of skill] (auth) {Current boundary $S\cap R_t\cap L$\\Authority may narrow\\Never silently widen};
\node[box,below=of evidence] (check) {Applicability and containment\\Reuse, reanalyze, or refuse};
\node[box,below=of auth] (commit) {Effect-owner commitment\\Fresh authority, identity, budget, nonce};
\node[box,below=of check] (receipt) {Observed occurrences $O_t$\\Outcome and lineage, including unknowns};
\draw[arr] (skill)--(auth); \draw[arr] (evidence)--(check);
\draw[arr] (auth)--(check);
\draw[arr] (check)--(commit);
\draw[arr] (commit)--(receipt);
\end{tikzpicture}
\caption{Proposed separation of authority, prospective evidence and observed effects. Only an admissible decision proceeds to effect-owner commitment; refusals do not commit effects. Evidence informs a decision but does not grant authority. This is not a claim that current AREG computes a complete prospective closure.}
\end{figure}
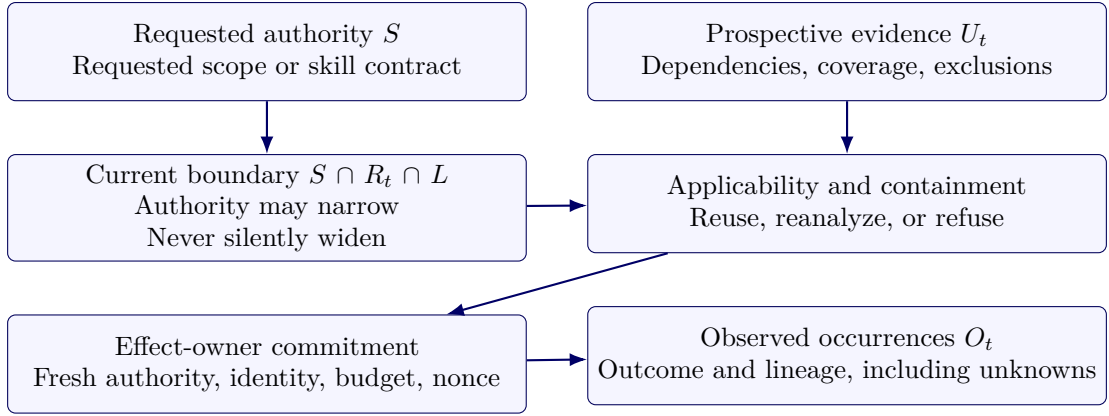

\section{Selective Revalidation}
\subsection{A coverage and dependency certificate}
Let $D$ be the set of effect-determining dependencies declared by a certificate $C_0$. Let $X|_D$ be the state projected onto those dependencies. A reusable certificate requires the conditional stability property
\begin{equation}
 X_t|_D = X_0|_D \ \land\ \Valid(C_0,t)
 \quad\Longrightarrow\quad
 \Tr(A,I,X_t;H) \subseteq U_0. \label{eq:stability}
\end{equation}
Here validity includes the analysis profile, exclusions, authenticated provenance and any externally checked freshness assumptions. Equation~\ref{eq:stability} is a requirement on the analysis and dependency model, not something a hash comparison proves. The mathematical premise is equality of projected values. Implementing it with digests additionally assumes faithful canonical encoding and resistance to relevant hash collisions; signatures require the usual unforgeability and key-custody assumptions. Changing an input outside $D$ is irrelevant only if this property is justified.

A proposed lease commits to $(A,I,L,C_0,X_0|_D)$ and to the issuer/audience, holder/session, policy and authority generations, lifetime, budget and consumption identity. This is an additional research commitment, not a change to CAVA's existing action identity. A whole-environment hash can remain a conservative fallback, but should not be called selective revalidation.

Implementation should separate an immutable grant from replaceable evidence. The grant fixes the formation-authorized target set, requested and granted envelope, and permitted state variation; a refreshed certificate references that same grant. Reanalysis can update evidence inside this delegation but cannot extend the grant, authorize a replacement executable or reinterpret exact-configuration approval as broader delegation.

\subsection{Validity is obligation-specific}
Effect dependencies alone cannot key a cached authorization decision. Write $D_{\mathrm{eff}}$ for the inputs on which the covered effect bound depends. Target authority, current runtime policy, issuer-key lifecycle, remaining quota and occurrence state have separate validity conditions. Equality on $D_{\mathrm{eff}}$ establishes at most the applicability of effect evidence; it says nothing about the other conditions.

A dependency certificate must therefore identify its \emph{claim}, not just a list of keys. For effect evidence, an appropriate abstract record is
\[
 C=(\mathit{claim},\mathit{analyzer},\mathit{language},I,H,
     D_{\mathrm{eff}},X_0|_{D_{\mathrm{eff}}},U_0,\mathit{validity}).
\]
Its authenticated reference points to the fixed grant. No actor-supplied ``verified'' bit substitutes for checking that reference and its supported profile. A certificate for a weaker language, different interpreter or incomplete observation horizon cannot be accepted because its effect tuples happen to look familiar.

\begin{table}[ht]
\centering\small
\begin{tabularx}{\linewidth}{@{}>{\raggedright\arraybackslash}p{0.25\linewidth}YY@{}}
\toprule
Obligation & Potentially reusable fact & Fact that must remain applicable at use \\
\midrule
Effect coverage & Result of supported analysis & Relevant state, workflow, target and coverage profile \\
Grant authenticity & Verification of immutable signed bytes & Trusted issuer and applicable key lifecycle \\
Authority & Original bounded delegation & Current restriction, revocation and applicable lifecycle \\
Quota and order & Declared ceilings and trace rules & Shared residual budget and committed prefix \\
Replay integrity & Stable occurrence identity & Durable consumption state and unresolved prior outcome \\
\bottomrule
\end{tabularx}
\caption{A cache hit is not an authorization decision. The prototype exercises a subset of these lifecycle conditions; it does not implement an enterprise key-rotation or distributed-revocation protocol.}
\label{tab:validity-domains}
\end{table}

Value equality also needs a declared meaning. For the closed DSL, an operand changing from $a$ to $b$ and back to $a$ restores the same instantaneous effect semantics if no relevant effect intervened. It does not restore a consumed occurrence, refund a spent budget, or prove that a previously revoked authorization has been reissued. History-sensitive obligations require durable history or monotone generations, not merely a hash of the final values. In the original straight-line prototype, reactivation or restoration of matching target/state fields can make an otherwise valid outstanding lease usable again. The active-state check enforces current eligibility, not permanent invalidation of an issuance generation. This restoration behavior is a source-level limitation, not an additional measured schedule or a revoke--regrant epoch protocol. In the distinct checked-summary store, restoring observations can similarly restore evidence applicability without restoring authority; forgetting a handle is not revocation of a grant.

\subsection{Decision procedure}
The initial fragment allows state variation inside a fixed, explicitly delegated executable and effect envelope. An approval that instead binds one exact configuration must not be treated as this broader delegation.
\begin{enumerate}
\item Validate action and formation-authorized target identity, current authority, expiry, revocation, unresolved prior outcomes and unconsumed invocation identity.
\item If coverage provenance or required dependency evidence is missing, refuse automatic execution. The reference outcomes distinguish review from explicit denial; neither commits an effect.
\item If Equation~\ref{eq:stability} remains applicable, reuse $U_0$. Otherwise obtain a fresh $U_t$ under the same supported coverage contract. No analyzer response means review, not an empty bound.
\item For a new invocation check $U_t\subseteq\Scope_t$, or an explicitly sound sufficient test. For an admitted continuation, compare its supported remaining bound with Equation~\ref{eq:residual}. Account for the remaining aggregate budget, not just individual effect tuples. Uncertainty that prevents establishing inclusion does not grant permission.
\item At each covered effect owner's commitment boundary, validate the applicable state/authority generation and atomically enforce scope, nonce and remaining budget. Record the resulting occurrence and receipt. Reconcile unknown outcomes before replay.
\end{enumerate}

The last step is not optional. A certificate may authorize preparation or admission without entitling a stale worker to commit later. Where a receiver cannot participate in generation fencing or fresh authority checks, the guarantee must weaken to an explicit revocation window or the path must remain ineligible for automatic reuse. Likewise, ``current state'' requires an observation contract: the analyzer and interpreter must agree on the snapshot used for the occurrence, and any mutable facts relied on at commitment must be checked in the owner's ordered state. Combining independently fresh but mutually inconsistent reads is not a coherent snapshot.

Budget enforcement here is an owner obligation, not a requirement to debit the
same cost twice. The original receiver charges at its atomic mutation boundary.
The later durable profile in Section~\ref{sec:durable} instead debits available
budget at admission and validates the retained reservation at claim. Fresh
authority and boundary checks remain necessary in either realization.

The following receiver proposition applies when authority updates and the actual mutation share that serialization order, as in the original SQLite effect owner. The GitHub pilot in Section~\ref{sec:saas-pilot} has a different boundary: it orders local authorization and intent immediately before transport dispatch, not atomically with GitHub's commit. A local revocation cannot recall an already sent request; no bound on the remote commit interval is established. Its observations therefore support a dispatch contract, not an instance of the receiver proposition for third-party storage.

\begin{proposition}[Conditional receiver non-widening]
Assume complete mediation of every covered effect, a sound check of the residual requested and leased trace boundaries, and binding of each checked concrete operation to the exact mutation committed. Applicable authority updates, shared-budget consumption, duplicate detection and commitment have a common serialization order. Then every committed covered trace prefix is permitted by the applicable skill request and issued lease and satisfies runtime authority at its commitment point.
\end{proposition}
\begin{proof}
Initially the committed history is empty. Complete mediation routes the next concrete mutation to the trusted effect owner. Its residual check admits only an extension inside $(S/h)\cap(L/h)$; the common serialization order accounts for prior effects, administrative updates and competing budget uses. The checked operation is the operation committed, so induction extends the invariant over the durable history. Runtime permission is checked at that extension's linearization point. This does not imply that a subsequent revocation retroactively makes prior effects nonexistent or that the whole historical trace satisfies a later policy revision.
\end{proof}

This receiver result does not establish prospective coverage. Equation~\ref{eq:coverage} is a separate analyzer obligation, defined on intended covered execution \emph{before} enforcement truncates it. Blocking an unpredicted effect can preserve receiver safety while falsifying the workflow-coverage claim. Neither result promises task completion or rolls back an already committed authorized prefix.

\begin{proposition}[Dependency-scoped reuse]
If Equation~\ref{eq:stability} holds and current boundary/lifecycle checks are retained, replacing fresh analysis by the applicable cached bound preserves the conditional coverage premise. This does not authorize an implementation outside the formation-authorized target set.
\end{proposition}
\begin{proof}
Equation~\ref{eq:stability} supplies the same required inclusion of possible traces in the selected bound. Boundary inclusion and commitment obligations remain unchanged. Target authority is a separate conjunct, so effect inclusion cannot discharge it.
\end{proof}

These propositions do not establish useful dependency extraction, precise bounds or low cost. A trivial bound can be sound and still refuse almost everything. The substantive research must show a justified, useful trade-off under real mappings and changing workflows.

\subsection{Operational decomposition and refusal ownership}
The algorithm has two distinct owners. Preparation may choose evidence; it cannot commit business effects. A participating effect owner may commit a concrete mutation only after checking its exact tuple and the current residual obligations. The following pseudocode is the receiver-serialized design contract rather than a claim that every listed enterprise lifecycle mechanism is implemented or that a local dispatcher owns a remote provider's transaction.
\begin{enumerate}[label=\textbf{P\arabic*.},leftmargin=3em]
\item Parse the request and verify its fixed grant, authorized target and delegation of state variation. Invalid credentials or a known incompatible target deny; unresolved required facts return review.
\item Validate the cached certificate's issuer/profile and its claim-specific applicability. A mismatch is a cache miss, not permission. If delegated analysis is unavailable, return review rather than an empty bound.
\item On a miss, derive a bound from the supported workflow and current snapshot. Keep the fixed grant unchanged. Reject an unsupported language or unresolved dependency; do not supply values from expected outcomes.
\item Check the bound against requested, leased and current authority with declared units, order and aggregate constraints. Issue an occurrence-specific preparation object only for a supported admissible bound.
\item At each effect, the owner serializes current-state checks, the exact effect match, residual quota/order and durable duplicate detection with the mutation. A successful preparation is not an exemption from this step.
\item Return committed, refused or unknown outcome separately. Resolve an unknown outcome against the original occurrence before considering another execution. A new authorization is an explicit event, not an automatic response to a failed cache lookup.
\end{enumerate}

The refusal path must name the fact it cannot establish. A policy violation is a denial; missing coverage may require review; a timeout after dispatch is an unknown outcome. These conditions must not all become a generic ``try again'' response. In particular, a certificate-valid but quota-exhausted request must not be routed to a faster lane that omits quota, and missing current authority cannot be repaired by a successful cached analysis.

\section{Cost Model and Break-Even Conditions}
\label{sec:cost-model}
\subsection{Account for checking the cache}
Let $A_{j,i}$ be the fresh-analysis cost on request $i$ in strategy $j$, $V_{s,i}$ the whole-snapshot lookup/validation and cache-maintenance cost, and $V_{d,i}$ the corresponding dependency-based cost. Let $M_{s,i}$ and $M_{d,i}$ indicate misses, and $P_s,P_d$ be actual initial preparation costs. In this additive work-accounting model, over $n$ requests,
\begin{align}
 T_f &= \sum_{i=1}^{n} A_{f,i},\\
 T_s &= P_s + \sum_{i=1}^{n}(V_{s,i}+M_{s,i}A_{s,i}),\\
 T_d &= P_d + \sum_{i=1}^{n}(V_{d,i}+M_{d,i}A_{d,i}).
\end{align}
Analysis and validation can be correlated with change type, input size and execution order; analysis timings need not be identical between strategies on the same input. These expressions do not replace conditional miss costs with an unconditional mean or treat a saved analyzer invocation as a fixed number of milliseconds. They exclude final authorization and commitment. Empirical whole-preparation timings include timer/runtime effects and are reported directly, not reconstructed by subtracting independent quantiles.

Dependency reuse improves this preparation objective over whole-snapshot caching exactly when
\begin{equation}
 \sum_i (M_{s,i}A_{s,i}-M_{d,i}A_{d,i})
 > (P_d-P_s)+\sum_i(V_{d,i}-V_{s,i}).
 \label{eq:break-even}
\end{equation}
This is an accounting identity, not a novel optimization theorem. It exposes the empirical question: are the avoided analyses worth the extra validation? In a constant-cost approximation with equal priming, the condition becomes $(h_d-h_s)A>V_d-V_s$, where $h$ is the hit fraction. For comparison against fresh analysis, priming also has to be amortized: $h_d A>V_d+P_d/n$. Neither condition follows merely from $h_d>h_s$.

\subsection{The measured implementation still traverses the input}
In the original straight-line implementation used for the cost study, let $W$ be workflow size, $N$ the size of the supplied state, $D$ the number of referenced keys and $E$ the emitted effect representation. The analyzer defensively copies and validates inert input before constructing effects. Ignoring bounded string costs and digest constants, its work is $O(N+W+E+D\log D)$. A snapshot lookup copies and hashes the workflow and state. The dependency lookup hashes the workflow and calls \texttt{dependencyDigest}, which itself validates and copies the \emph{entire supplied state} before selecting keys.

Consequently, the current dependency hit path is not $O(D)$; it retains an $O(N+W)$ traversal plus dependency selection and hashing. Relevant changes add the miss validation cost to fresh analysis. This is a concrete reason to expect regressions for small workflows with large state, rather than a measurement artifact to remove from the report. The limits of 64 leaves and four queue levels bound this implementation's supported language, not a general constant-time claim.

Trusted immutable snapshot handles or authenticated change indexes might avoid full traversal. They would move work to snapshot construction and update maintenance and introduce new trust and invalidation obligations. Such costs cannot disappear from the denominator, and such mechanisms are not implemented by relabeling the existing cache.

\subsection{What a faster preparation does not imply}
Under a simplified additive model with unchanged remaining cost $G$, an improvement from preparation $A$ to $A'$ gives total speedup $(A+G)/(A'+G)$. Even eliminating all preparation cannot exceed $(A+G)/G$ in that model. Real queueing, network and scheduling effects can violate the additive approximation; whole-workflow measurements are still required. Quantiles cannot be subtracted to estimate stage costs, and closed-loop per-request timings do not establish throughput under an independently scheduled offered load.

\section{Counterexamples and Scope Limits}
\paragraph{An omitted dependency.}
Consider two executions with identical visible action, implementation, certificate and declared dependency digests. A hidden feature flag makes one execution export customer data. A decision procedure seeing only these inputs cannot distinguish them. Authenticating their identical certificate does not repair the missing input. This is a standard indistinguishability argument, closely related to the insufficiency of invocation-only evidence; it is a limitation to test, not our claimed primary theorem. Per-effect authorization at a genuinely mediated export receiver can still block the export, but the earlier claim to have covered the entire workflow is false.

\paragraph{A race after admission.}
At $t_0$ the gate validates state generation $g$. At $t_1$ a configuration update or revocation changes the applicable generation. At $t_2$ an external receiver performs the effect. A check only at $t_0$ cannot guarantee a $t_2$ boundary. The receiver needs an ordered fence or fresh applicable authority. A changed generation may conservatively require renewed admission even if the new effect would ultimately be safe.

\begin{figure}[ht]
\centering\small
\begin{tabularx}{\linewidth}{@{}l*{3}{>{\raggedright\arraybackslash}X}@{}}
\toprule
Path & Preparation/admission & Intervening event & Effect-owner commitment \\
\midrule
Admission only & Validate $g$, issue receipt & State changes to $g+1$ & Stale work may execute \\
Fenced path & Validate $g$, reserve scope & State changes to $g+1$ & Refuse stale generation; re-evaluate \\
Stable fast path & Reuse applicable bound; check authority & No relevant change & Consume once within current scope \\
\bottomrule
\end{tabularx}
\caption{Ordering, not a measured latency chart. ``Fast'' removes redundant analysis, not revocation or effect-owner checks. A receipt attached after an effect is not a preventive fence.}
\end{figure}

\paragraph{Partial execution and asynchronous composition.}
A workflow can commit an authorized refund before a later export is denied. Denying the export does not undo the refund; compensation is an independent policy-controlled action. A downstream job, webhook or second agent requires explicit delegation, occurrence identity and a covered receiver. A union of single-invocation graphs does not establish authorization for AND-joins, cycles, shared counters or unbounded recursion. The first prototype excludes these general cases and bounds its horizon.

\paragraph{Source and binary tools.}
The method does not assume every tool is available as source, but it also does not infer binary compatibility. Exact artifact/config pinning plus receiver-mediated effect constraints may provide a narrower containment path. Remote SaaS internals and an unverifiable vendor attestation cannot be advertised as independently computed closure coverage.

\section{Finite Model and Counterexample Checks}
The accompanying Node model contains 35 named states and four strategies. Fixtures and expected outcomes were written before each corresponding evaluator change by the same engineering effort. They are not independent annotations or held-out real workflows. Authentication, revocation, dependency completeness and final fencing are abstract Boolean/state inputs. No network, signature verification, program analysis, database transaction or external side effect occurs.

The primary fixtures contain six expected allows and twenty-nine non-allows, including eight reviews. Identity-only is an intentionally incomplete ablation. Exact-snapshot invalidation conservatively refuses changed snapshots. Fresh analysis and dependency-scoped reuse share the represented authority, state-variation delegation, identity, budget and final-fence obligations; the latter substitutes an applicable bound when possible. One warm analyzer-outage case intentionally distinguishes their availability.

\begin{table}[ht]
\centering\small
\begin{tabular}{@{}lrrrr@{}}
\toprule
Strategy & Admits & Against label/29 & Refused admissible/6 & Analysis attempts \\
\midrule
Identity-only ablation & 14 & 8 & 0 & 0 \\
Exact snapshot & 3 & 0 & 3 & 0 \\
Fresh analysis & 5 & 0 & 1 & 22 \\
Dependency-scoped & 6 & 0 & 0 & 6 \\
\bottomrule
\end{tabular}
\caption{Outputs of the abstract development model. Non-admission labels include deny and review, not only demonstrated harmful effects. Attempts include unavailable-summary branches; 20 fresh and 5 selective attempts select available summaries. No real analyzer, latency or CPU cost is measured.}
\label{tab:model}
\end{table}

The candidate matches all 35 fixture decisions, and 45 software tests pass after adversarial code-review repairs. The expanded tests first exposed ten failures, including absent identity/lifecycle facts, undelegated configuration variation, and missing commit-time authority checks. The final model explicitly checks post-admission runtime narrowing, expiry and nonce consumption, and includes a two-occurrence aggregate-budget case. These results are consistency checks, not evidence that an analyzer generalizes. The intended distinctions are that unrelated changes need not invalidate a bound, narrowing does not expand authority, and post-check changes require the final fence.

One additional broken-premise fixture is reported separately and is \emph{admitted} by the candidate while its declared external outcome includes an unauthorized export. Its hidden input is absent from the dependency list and both supplied summaries. The model marks commitment of the supplied bound, not per-effect inspection of an actual receiver; it therefore cannot establish safety for the hidden effect. Retaining this example makes the coverage assumption visible and motivates independent dependency discovery and mandatory receiver mediation. It must not be hidden inside a ``100\% safe'' headline.

\section{Original Straight-Line Implementation}
\label{sec:prototype}
The original Matoi prototype realizes the applicability contract through one narrow straight-line executable semantics. Figure~\ref{fig:matoi-path} depicts this implementation and its receiver-evidence bridge, not every later experiment. The macOS/Azure campaigns, receiver-evidence campaign and preparation-cost measurements use this implementation; they do not use the branch/collection kernel or checked ingestion introduced in Section~\ref{sec:checked-semantics}. Its unit of reuse is a supported consequence analysis; its unit of enforcement is a concrete occurrence at an effect owner. These units must remain distinct even when they share a canonical action identity.

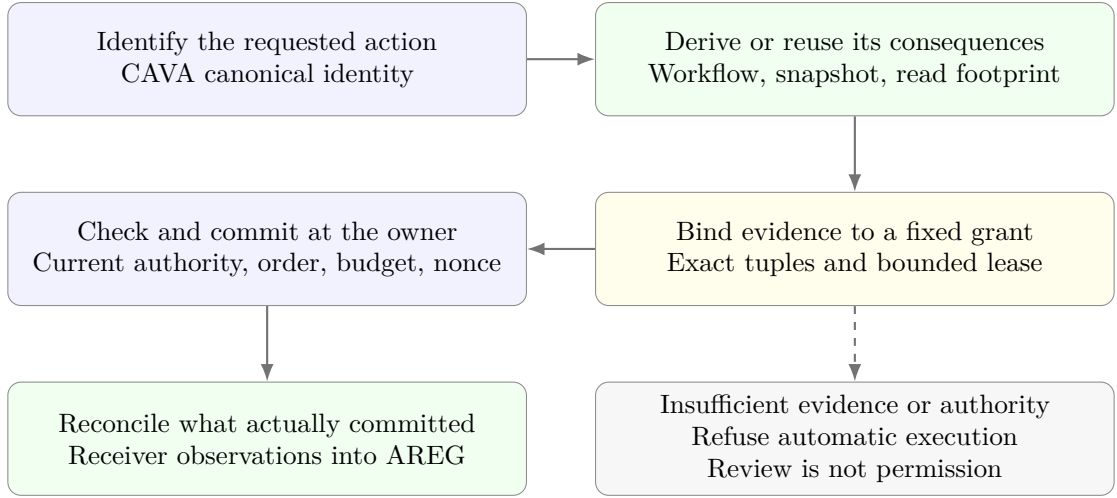
\begin{figure}[ht]
\centering
\begin{tikzpicture}[box/.style={draw=black!45,rounded corners=5pt,align=center,text width=0.40\linewidth,minimum height=15mm,font=\small},arr/.style={-{Latex},thick,draw=black!55},node distance=10mm and 9mm]
\node[box,fill=blue!5] (action) {Identify the requested action\\CAVA canonical identity};
\node[box,fill=green!6,right=of action] (derive) {Derive or reuse its consequences\\Workflow, snapshot, read footprint};
\node[box,fill=yellow!9,below=of derive] (lease) {Bind evidence to a fixed grant\\Exact tuples and bounded lease};
\node[box,fill=blue!5,left=of lease] (owner) {Check and commit at the owner\\Current authority, order, budget, nonce};
\node[box,fill=green!6,below=of owner] (evidence) {Reconcile what actually committed\\Receiver observations into AREG};
\node[box,fill=black!3,right=of evidence] (refuse) {Insufficient evidence or authority\\Refuse automatic execution\\Review is not permission};
\draw[arr] (action)--(derive);
\draw[arr] (derive)--(lease);
\draw[arr] (lease)--(owner);
\draw[arr] (owner)--(evidence);
\draw[arr,dashed] (lease)--(refuse);
\end{tikzpicture}
\caption{The Matoi reference path. Actual experiments use CAVA canonicalization, the BAF tuple matcher and the AREG graph builder. Signed grants and receipts instantiate the authority/proof role; this is not a claim that every PCAA service or production ZeroGate integration was exercised. Dashed flow denotes refusal without a new effect.}
\label{fig:matoi-path}
\end{figure}

\subsection{Closed workflow semantics and a read-footprint lemma}
The prototype accepts a versioned JSON language with refund, notification, ordered sequence and bounded queued-child operations. Values are literals or named state references. Refunds use positive integer US-dollar cents; notification destinations are exact strings. Unknown operations, missing references, malformed input, more than 64 leaves and more than four nested queues produce review. The accepted language has no arbitrary code, conditional hidden callbacks, state-writing expressions or undisclosed service internals.

Let $V_x(e)$ evaluate an expression against an immutable snapshot $x$. A literal reads no state; a reference reads its named key. Let $D(w)$ be the union of keys in all operand expressions of workflow $w$, including queued descendants. Let $F_x(w)$ concatenate the resulting concrete effect tuples in execution order, flattening queued children at the point where the parent waits for their completion. Enqueueing itself is a scheduling event, not a business effect. The interpreter and analyzer are separately implemented, but share this declared language specification.

\begin{proposition}[Read-footprint stability in the closed fragment]
For an accepted workflow $w$ and valid snapshots $x,y$, if $x|_{D(w)}=y|_{D(w)}$, then $F_x(w)=F_y(w)$.
\end{proposition}
\begin{proof}
Literal evaluation is state independent. A reference has the same value under equal projected snapshots. Therefore every operand of a leaf has the same value, so its effect tuple is unchanged. Sequence preserves equality by concatenation. The queued-child case applies the same argument recursively and includes the child's footprint in the parent's union. Structural induction completes the argument. Rejected syntax has no such derivation and is not interpreted as an empty trace.
\end{proof}

For this straight-line fragment, the prospective analyzer constructs $F_x(w)$ exactly under the declared snapshot semantics; it is not discovering effects of arbitrary programs. A separately implemented actor evaluates operands and submits concrete operations, rather than executing a supplied list of predicted effects. Both use the same immutable snapshot for that invocation. The receiver compares signed workflow and dependency commitments against its current state before each effect. A relevant mismatch observed at commitment stops the stale suffix. These current-value checks do not detect a change followed by restoration between observations and do not provide monotone issuance-generation invalidation.

\subsection{Effect-owner serialization and actual mutations}
The trusted controller records a grant and current runtime state, derives evidence and signs a per-occurrence lease using Ed25519. The action fingerprint uses CAVA canonicalization. A separate actor process sends operations over HTTP; an enqueue starts another process. MCP-style and SDK-style inputs use two local request encodings. These are not claims that external MCP servers or vendor SDK distributions were integrated in this experiment.

The effect owner verifies the signature, current executable digest commitment, workflow and relevant-state commitments, expiry, current authority, the exact indexed effect and its original/leased/current tuple boundaries. Tuple checks invoke the existing BAF execution-boundary matcher. Within one SQLite \texttt{BEGIN IMMEDIATE} transaction, the owner enforces sequence, idempotency and a group-wide integer-cent budget and writes the concrete synthetic ledger or notification effect with its receipt. Administrative updates use that same database and serialization order. The checked tuple is the tuple stored as the effect: there is no subsequent unfenced external payment API in the claimed boundary.

Re-delivery of an already committed index returns its prior receipt without another mutation. A deliberately dropped HTTP acknowledgment is reconciled through receipt lookup, not blind re-execution. Revocation between the refund and its queued notification preserves the committed refund and rejects the notification. This exposes partial completion rather than claiming rollback.

The controller, signing key, selected interpreter and owner are trusted. Executable digests are computed from actual source/runtime bytes, but a digest is not remote attestation or a hostile-host sandbox. The lease is a bearer credential within this isolated test transport; production holder authentication and TLS are not established by the loopback HTTP experiment. SQLite durability is configured with WAL and \texttt{synchronous=FULL}; abrupt power-loss recovery is not inferred from ordinary successful transactions.

\subsection{Matched cached comparisons and measurement boundaries}
All three arms use the same analyzer, signed lease format, actor, effect owner, final checks and storage. Always-fresh recomputes evidence on each request. Whole-snapshot caching reuses only an equal full snapshot and refreshes otherwise. Dependency caching reuses an applicable workflow/read-footprint certificate and refreshes on relevant changes. Both caches retain mutable authority checks. Unlike the earlier exact-refusal model, the whole-snapshot primary comparator refreshes evidence after invalidation instead of permanently refusing safe changed workflows.

The development matrix fixes 21 schedules before execution: cold and warm operation, irrelevant and relevant drift, out-of-bound amounts/destinations, unknown/missing input, commit-time revocation/narrowing/target or dependency change, partial execution, lost acknowledgments, duplicate delivery, concurrent budget consumption, invalid signature, actor-effect drift, expiry and downstream expansion. Repetitions and two encodings do not create new independent skill families. No existing collaboration holdout is used.

Per-request selection, signing and elapsed durations are recorded separately from cache-priming setup. Elapsed time includes actor/queue process creation, HTTP and receiver commitment after registration; it does not include provisioning, cache priming or registration. A full lifecycle cost must include these separately reported stages. There is no injected analyzer delay to manufacture a speedup. An independent code path replays the controller journal and verifies concrete commitments, signatures, state changes, budgets and receipts; this is internal cross-checking, not an external party's validation. Source hashes include both the prototype and imported CAVA code.

\subsection{Local real-process results}
The measured macOS run uses Node 22.18.0 and ten repetitions of each schedule/arm/encoding combination. All 154 unit, integration and verifier tests pass. Its 1,260 schedule records include 60 two-occurrence budget schedules, producing 1,320 workflow occurrences: 480 fully complete, 60 commit an authorized prefix before refusal, and 780 commit no effect. The owner records 1,020 committed effects. All declared schedule outcomes match after replay; this is agreement with authored development expectations, not an independently labeled false-allow-rate estimate.

Each arm has 420 schedule records, 340 committed effects and 160 fully completed occurrences when concurrent peers are included. Primary-worker completion alone differs because either concurrent worker may win the shared budget. Reporting that arbitrary winner as accuracy would be misleading. In all 60 concurrent schedules, exactly one 7,300-cent refund consumes the shared 10,000-cent budget; the other is refused. The 60 lost-acknowledgment schedules recover their existing receipts without duplicate effects. The 60 partial-revocation schedules retain the refund but block the queued notification.

\begin{table}[ht]
\centering\small
\begin{tabular}{lrrrr}
\toprule
Arm & Request analyses & Including priming & Effects & Elapsed p95 (ms)\\
\midrule
Always-fresh & 420 & 820 & 340 & 78.161\\
Whole-snapshot cache & 160 & 560 & 340 & 76.973\\
Dependency cache & 140 & 540 & 340 & 77.450\\
\bottomrule
\end{tabular}
\caption{Local development measurements, not a throughput or SLA test. Each arm includes 400 separately timed cache-priming analyses, including unused priming in always-fresh. Elapsed timing begins after priming and registration and mixes completed and refused schedules; stratified values are retained in the artifact.}
\label{tab:localprototype}
\end{table}

Relative to whole-snapshot caching, dependency reuse eliminates 20 of 160 request-time analyses (12.5\%), or 20 of 560 analyses including priming (3.6\%). Against always-fresh the corresponding reductions are 66.7\% and 34.1\%. These different denominators must not be substituted for one another. Overall elapsed p95 is slightly higher than the cached comparator, and no statistically established end-to-end improvement is claimed.

For irrelevant-state drift specifically, selection p95 is 0.040 ms with dependency caching versus 0.079 ms with whole-snapshot caching; elapsed p95 is 77.442 versus 76.675 ms. The tiny analysis savings do not dominate process creation, HTTP and durable commitment. No artificial analyzer delay was inserted. A useful production advantage would require evidence from workloads where justified analysis is material, while retaining the same security obligations and accounting for all preparation costs.

\subsection{Azure Linux reproduction}
The same source and schedule matrix was rerun on an Azure Standard\_D4s\_v3 VM in Central US, zone 1, using Ubuntu 24.04 and Node 22.18.0 (x86-64, Intel Xeon E5-2673 v4 at 2.30 GHz). The initial D2s\_v5 allocation was unavailable and produced no study observations. The replacement VM ran the same 154 tests successfully, followed by 1,260 schedule records and 1,020 committed effects. Full completion, partial completion and zero-effect counts match the local run; which concurrent peer wins can differ.

Linux elapsed p95 was 241.079 ms for always-fresh, 242.746 ms for whole-snapshot caching and 238.182 ms for dependency caching. However, dependency caching's p99 was 258.589 ms versus 255.759 ms for the cached comparator, and summed elapsed time was also slightly higher. These descriptive values do not support a uniform or statistically established speed advantage. They strengthen reproducibility of the stated boundary behavior, not a throughput or production-latency claim.

The source archive and result/database archive were SHA-256 checked after download, and the cloud evidence was replayed on the original machine. Replay uses the recorded Linux interpreter digest without pretending that the macOS binary attests it. Both campaigns together contain 2,520 schedule records and 2,040 synthetic effects; they repeat the same authored family and are not independent domain-validation datasets. Services remained bound to loopback; only SSH from the operator's current egress IP was exposed on the temporary VM.

\section{Receiver-Bound Evidence and Reconciliation}
\label{sec:evidence-bridge}
The prospective list and an exposure graph are not interchangeable evidence.
We implement a Matoi--AREG evidence bridge that projects signed expected
occurrences separately from receiver-confirmed occurrences, using the existing
AREG graph builder rather than a substitute visualization. The bridge tests
an occurrence-evidence interface; it does not infer previously unobserved
consequences or replace the effect owner's enforcement.

\subsection{Confirmation requires a bound receiver observation}
For an effect at index $i$, the bridge's matching key is
\[
 K_i=(\mathit{run},\mathit{root},H(L),i,e_i).
\]
Equality of the entry action alone cannot discharge this match. A trusted
collector opens the stopped owner's SQLite database read-only, reads the
ordered commitment journal and effect table, and signs a domain-separated
snapshot containing the run, root, exact lease digest, owner-source digest,
database digest, journal high-water mark, capture time and evidence. The
issuer and collector public keys, expected identities and source commitment
come from the verifier's trust input, not from actor-controlled packet fields.

The bridge checks the issuer signature, collector signature, snapshot bindings,
ordered journal replay and exact readback before creating a confirmed node.
The separately invoked archive verifier additionally rehashes saved files and
re-reads the database. A manifest digest must be obtained outside the mutable
bundle. Neither a colocated checksum nor an actor's \texttt{verified} flag is
a trust anchor. A signature authenticates the collector's claim; it does not
attest host integrity or prove that no uninstrumented effect occurred.

Repeated delivery of one $K_i$ yields one confirmed occurrence. Identical
effects at distinct indices remain separate occurrences. A new lease cannot
claim an old lease's receipt solely because the root, index and effect match.
Full typed tuples and lease membership are retained as sidecars. The current
owner does not attest physical process ancestry, so the projection creates
no inferred \texttt{causal\_parent} edges from that membership.

\subsection{Snapshot outcomes and verification}
The outcomes distinguish all signed indices confirmed, partial confirmation,
none confirmed, an empty expectation, and unknown evidence. None confirmed
does not mean no hidden action occurred; all expected indices confirmed does
not establish future containment or general task completion. A denied attempt
is recorded separately, never counted as an executed effect. Its reason is
owner-reported: the frozen replay does not independently derive refusal reasons.
AREG's coverage remains explicitly incomplete and containment unestablished.

The development suite passes 45 tests, including a 48-run process integration
test over eight authored schedules, two local encodings and three repetitions.
Its temporary fixtures are test evidence, not an additional retained empirical
corpus. The suite tests partial
revocation, duplicate delivery, lost acknowledgments, cross-lease substitution,
forged or missing evidence, live-WAL rejection and six saved-artifact corruption
classes. Two review findings were reproduced and repaired: wide JSON could
allocate too much before its structural limit, and a coerced non-string digest
could disagree with AREG's node validation. A bounded-heap regression and strict
digest/graph-count checks now cover these cases. These checks support the
implemented evidence interface, not prospective coverage of arbitrary programs
or independently labeled authorization accuracy. Previous measured campaigns
and collaboration hold-outs are unchanged.

A subsequent, separately retained campaign ran the fixed 48-schedule matrix
and was accepted by a separate archive-verifier invocation with an externally
recorded manifest digest. Of 96 signed expected occurrences, the receiver
confirmed 54 synthetic effects. The snapshots contain 24 all-expected-confirmed,
six partial, and 18 none-confirmed outcomes; there were 24 owner-reported denied
attempts and six duplicate-delivery reconciliations, which add no effects.
Normal operation, duplicate delivery, lost acknowledgments and identical leaves
confirmed both effects. Partial revocation confirmed the first and refused the
second; initial revocation, effect drift and relevant drift confirmed neither.
Each schedule is repeated six times across the two encodings and three rounds.
No record failed replay and none had unknown evidence in this valid-snapshot
matrix; malformed and missing observations are exercised separately in tests.
These counts are agreement with authored expectations, not an independently
estimated false-allow rate. A final combined regression passed 223 tests,
comprising the 45 new tests and 178 existing Matoi/AREG tests. The local archive
and digest are identified in the accompanying bridge report; this is not a
new public download, external replication, or production deployment.

\section{Preparation Cost Boundary Study}
\label{sec:cost-study}
\subsection{Question, protocol and measured boundary}
The real-process study establishes a useful negative observation: saved analyzer calls did not establish a total-workflow speedup. We therefore isolate a smaller empirical question: how much does evidence preparation itself cost as the supported workflow and supplied state grow? This second study does not remeasure security or remove live checks from an end-to-end comparator. It measures a computation that all three execution designs perform before the unchanged owner obligations.

The protocol fixes a Cartesian grid of 1, 8, 32 and 64 workflow leaves; 0, 32, 512 and 4,096 base unreferenced state keys; and four change classes, giving 64 cells. Each leaf reads a destination, and alternating refund leaves also read an integer-cent amount. Notification text is a fixed literal. Thus read-key count is $k+\lceil k/2\rceil$ for $k$ leaves; it is not held constant across workflow sizes. The irrelevant-change class adds one unreferenced mutation key, including at base size zero. The archive records actual total and unrelated key counts, not just the grid labels.

Within a cell, two excluded warmup epochs precede twelve measured epochs. Each epoch uses newly created caches and sixteen sequential requests per arm, with the same deterministic inputs for all arms. Arm order rotates with the epoch. Stable requests do not change. Every irrelevant-state request changes the extra unread value; every relevant-state request changes a read refund amount. Workflow-content requests vary zero to four enclosing queued-child levels, retaining the same ordered leaf effects but changing workflow bytes between adjacent requests. These are controlled pure change streams, not observed production change frequencies or a sweep over drift cadence.

The timer includes lookup, actual analysis when required, key construction, cache maintenance and return. Both caches pay real initial preparation. Always-fresh pays no unused priming, improving the accounting over the earlier prototype campaign. The study imports the frozen analyzer and digest functions; its small selector reuses a snapshot key already computed during lookup instead of hashing the entire snapshot a second time on a miss. It is differentially checked against the frozen selector. This avoids manufacturing an advantage against a needlessly inefficient cached comparator.

Input construction, evidence hashing, fresh-result comparison, file writes and final owner execution are outside the timer. After every timed request, the selected result is compared with fresh analysis and the frozen selector. Those checks can influence later allocation, JIT and garbage-collection behavior even though their own durations are excluded. The host is not CPU-pinned or thermally controlled. These conditions limit latency inference and are retained in the environment record.

\subsection{Predictions and accounting}
The structural prediction is not that dependency reuse always wins. A stable stream permits both caches to reuse; an irrelevant-change stream distinguishes their invalidation rules. Relevant changes force both to analyze, while workflow-content changes force reanalysis even though the selected test transformations preserve effects. The latter is a deliberately conservative limitation: this implementation does not prove arbitrary workflow rewrites equivalent.

At sixteen requests per measured epoch, fresh analysis performs sixteen timed analyses and no prime. Each cached arm performs one prime. A hit-only stream therefore needs one analysis including priming; a miss-only stream needs seventeen. Dependency caching is hit-only for stable and irrelevant-state streams. Whole-snapshot caching is hit-only only for stable streams. These call counts follow from the protocol, whereas elapsed cost must be measured. Full-state validation can dominate even a successful dependency hit.

Primary cost is the measured sum of initial preparation and sixteen request preparations, with paired comparisons within the same cell and epoch. Nearest-rank quantiles and all raw timings are retained. Repeated requests are correlated observations, not independent workflow examples. No artificial analyzer delay, adaptive stopping, removed outlier, or successful-request-only denominator is used. The separately coded verifier reconstructs inputs, counts, intervals and summaries without importing the study runner. It still uses the same analyzer for semantic consistency; it is neither an independent ground-truth oracle nor proof that an unsigned timing history could not be fabricated.

\subsection{Measured response surface}
The fixed local campaign runs on an Apple M4 Pro workstation with 24 GiB memory, macOS/arm64 and Node 22.18.0. It records 896 epochs: 128 warmup and 768 measured. The measured epochs contain 36,864 preparation requests, with 6,144 additional warmup requests retained but excluded from summaries. All recorded requests agree with the fresh analyzer and frozen selector, and the archive verifier accepts the complete matrix. Twelve test groups cover protocol construction, real helper invocation accounting, fresh/cached selection, warmup exclusion, and rejection of missing, modified or internally inconsistent evidence. These are internal checks, not an independent authorization-accuracy estimate.

Each arm has 12,288 measured requests. The 36,864-request total excludes the 1,536 initial cached-arm preparations, whose costs are nevertheless included in measured epoch totals. Always-fresh performs 12,288 analyses. Whole-snapshot caching performs 9,984 including 768 primes; dependency caching performs 6,912 including 768 primes. The reductions against these two comparators are therefore 43.75\% and 30.77\%, respectively, for \emph{analysis calls including priming}, not elapsed time. Their different hit patterns do not remove the cost of validating a hit.

\begin{table}[ht]
\centering\small
\begin{tabular}{@{}rrlrrr@{}}
\toprule
Leaves & Unread keys & Change & Fresh & Snapshot & Dependency \\
\midrule
1 & 0 & Stable & 11.096 & 10.724 & 9.502 \\
1 & 1 & Irrelevant & 10.992 & 22.896 & 9.610 \\
1 & 0 & Relevant & 10.315 & 20.909 & 18.633 \\
1 & 0 & Workflow & 14.586 & 27.169 & 20.265 \\
1 & 4,096 & Stable & 1,367.917 & 4,760.271 & 1,486.138 \\
1 & 4,097 & Irrelevant & 1,398.182 & 6,050.510 & 1,499.318 \\
1 & 4,096 & Relevant & 1,390.138 & 6,151.599 & 2,818.651 \\
1 & 4,096 & Workflow & 1,423.495 & 6,110.838 & 1,522.021 \\
\midrule
64 & 0 & Stable & 345.235 & 354.495 & 260.169 \\
64 & 1 & Irrelevant & 350.047 & 748.521 & 257.878 \\
64 & 0 & Relevant & 368.107 & 771.857 & 633.612 \\
64 & 0 & Workflow & 382.250 & 776.000 & 548.797 \\
64 & 4,096 & Stable & 1,857.883 & 5,226.529 & 1,873.596 \\
64 & 4,097 & Irrelevant & 1,884.852 & 7,213.555 & 1,862.463 \\
64 & 4,096 & Relevant & 1,854.703 & 7,238.229 & 3,421.529 \\
64 & 4,096 & Workflow & 1,868.378 & 7,226.711 & 2,108.143 \\
\bottomrule
\end{tabular}
\caption{Preparation microseconds per request: median of twelve epoch totals, each divided by sixteen and including priming. All four change classes at the smallest/largest workflow and base-state corners are shown; the artifact includes all 64 cells. These are not end-to-end durations or independent-quantile differences.}
\label{tab:cost-corners}
\end{table}

Dependency caching's median paired cost is lower than whole-snapshot caching in all 64 cells, but lower than fresh analysis in only 23. In stable and irrelevant-state streams, it is cheaper than fresh analysis in 11/16 and 12/16 cells, respectively. In neither relevant-state nor workflow-content streams does it have a lower median paired cost than fresh analysis. Across individual paired epochs, it is cheaper than fresh in 133/192 stable, 137/192 irrelevant, 0/192 relevant and 1/192 workflow comparisons. These frequencies describe this fixed balanced campaign, not a production workload distribution or statistical significance.

For 64 leaves and one unread mutation key, the median paired dependency/fresh ratio is 0.741 and dependency/snapshot ratio is 0.343: approximately 25.9\% and 65.7\% lower preparation cost. This is a useful reuse regime. At one leaf and 4,096 unread keys with no drift, dependency/fresh is instead 1.071 despite a cache hit on every request. For relevant drift at 64 leaves and 4,096 unread keys, it is 1.844. These losses are consistent with the full-state validation and miss costs identified in Section~\ref{sec:cost-model}; fewer analyzer calls do not make those traversals free.

\begin{figure}[ht]
\centering
\begin{tikzpicture}[x=3.1cm,y=1cm,font=\small]
\draw[draw=black!45,-{Latex}] (0,0)--(2.1,0);
\foreach \x/\label in {0/0,0.5/0.5,1/1.0,1.5/1.5,2/2.0}{
  \draw[black!45] (\x,0.05)--(\x,-0.05);
  \node[below] at (\x,-0.05) {\label};
}
\draw[black!50,dashed] (1,0.1)--(1,3.35);
\node[above,align=center] at (1,3.35) {Fresh-analysis cost};
\filldraw[draw=green!40!black,fill=green!12,rounded corners=2pt] (0,2.5) rectangle (0.741,3.05);
\filldraw[draw=blue!50!black,fill=blue!10,rounded corners=2pt] (0,1.45) rectangle (1.071,2.0);
\filldraw[draw=orange!60!black,fill=orange!12,rounded corners=2pt] (0,0.4) rectangle (1.844,0.95);
\node[anchor=east,align=right] at (-0.08,2.775) {64 leaves, irrelevant drift\\one unread key};
\node[anchor=east,align=right] at (-0.08,1.725) {One leaf, no drift\\4,096 unread keys};
\node[anchor=east,align=right] at (-0.08,0.675) {64 leaves, relevant drift\\4,096 unread keys};
\node[anchor=west,fill=white,inner sep=1pt] at (0.77,2.775) {0.741};
\node[anchor=west,fill=white,inner sep=1pt] at (1.10,1.725) {1.071};
\node[anchor=west,fill=white,inner sep=1pt] at (1.86,0.675) {1.844};
\node[below,align=center] at (1,-0.65) {Dependency / fresh preparation cost\\Lower is better; includes initial preparation};
\end{tikzpicture}
\caption{Three diagnostic regimes from the fixed 64-cell study, shown as median paired epoch-cost ratios. A useful hit saves work, an expensive hit can lose, and a relevant-change miss can pay twice. These selected cells explain the mechanism; they are neither an overall average nor end-to-end speedups.}
\label{fig:cost-regimes}
\end{figure}
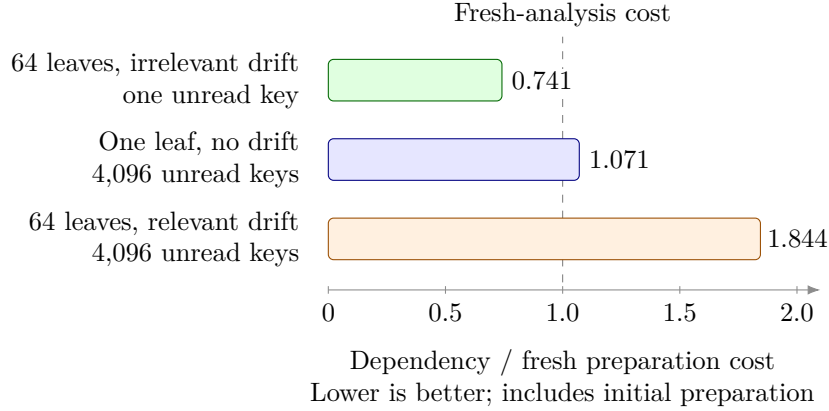

\subsection{What the result changes}
The result favors neither unconditional fresh analysis nor unconditional reuse. It motivates choosing a preparation strategy according to measured validation cost, workload size and change cadence, under unchanged authorization obligations. An adaptive selector would require its own frozen training/evaluation split and switching-cost accounting; none is claimed here. A trusted immutable snapshot representation could also avoid repeated full-state validation, but its construction, update maintenance, storage and invalidation costs would have to be measured rather than moved outside the reported boundary.

The original process campaigns and this preparation study answer different questions. The process campaigns retain real signing, current authority and receiver commitment but establish neither a consistent elapsed-time advantage nor a full-lifecycle speedup. The preparation study explains a trade-off but does not execute those final checks. It also does not measure the distinct checked-ingestion traversal below. Multiplying a preparation ratio into an earlier workflow percentile or treating saved GitHub preparation calls as this study's latency savings would be invalid. The evidence supports a bounded design refinement and further matched testing, not a faster production ZeroGate release.

\section{Checked Query-Aware Evidence}
\label{sec:checked-semantics}
\subsection{A distinct finite semantics}
The query-aware kernel is a separate implementation, not an extension silently
substituted into the earlier measurements. Its bounded JSON language supports
symbolic publication, ordered sequences, Boolean branches, ordered collection
iteration, keyed lookup with explicit absence, nested item scopes and queued
child traces. A publication denotes a destination and positive integer byte
count; the semantic kernel transmits no content. A queue denotes an ordered
child trace, not asynchronous process execution. The accepted profile limits
workflow nodes to 256, collection membership to 16, expanded visits to 4,096,
effects to 64 and nested queues to four, with additional parsing, depth and
serialization limits. Unknown syntax, missing data and resource overruns return
review rather than an empty successful trace; a valid empty sequence is distinct.

Unlike the straight-line key set $D(w)$, the active footprint depends on the
execution selected by the snapshot. Write $Q_x(w)$ for its finite map from
queries to observations and $F^q_x(w)$ for its ordered effect trace. Scalar
queries record predicate inputs and operands. Membership queries record the
ordered member identifiers, including an empty result; item queries include
collection, identifier and field, preserving lexical scope. Keyed lookup
records presence or absence of the selected identifier. Thus adding a previously
absent callback invalidates its negative lookup even if every previously read
item is unchanged. Inserting an unrelated key need not invalidate a negative
lookup, whereas insertion or reordering in an iterated collection changes its
membership observation. Unused metadata is not an effect dependency.

The recursive analyzer constructs this exact trace and footprint. A separately
coded stack interpreter evaluates workflow and snapshot inputs, not predicted
effects. Both share parsing, structural validation and primitive type checks.
Exact observation equality is deliberately conservative: changing a selected
operand can cause non-reuse even if the changed effect would still be authorized.
This is a strong query-dependency baseline, not a more permissive semantic-region
algorithm or a discovery procedure for arbitrary program behavior.
This legacy workflow commitment hashes raw UTF-8 source text under a
well-formed-encoding assumption, so formatting-only rewrites invalidate reuse.
It is not the canonical parsed-workflow identity of the later portable profile
(Section~\ref{sec:portable-update}).

\subsection{Checking completeness before retaining evidence}
The original trusted-summary API compares the observations that a caller
supplies; completeness and truthful effects are its preconditions. A summary
with a branch dependency removed can therefore pass while the interpreter
produces an unrepresented external-destination effect. The retained
counterexample violates that trust premise; it does not refute a complete,
protected dependency cache.

Checked ingestion treats the candidate summary as data. A reference traversal
of the issuance workflow and snapshot, implemented by the stack interpreter's
evidence path, reconstructs all active observations and ordered effects without
calling the recursive analyzer or consuming the candidate trace. Ingestion
requires identity/profile validity, exact equality of the query-observation
map, and equality of the ordered effect list. Dependency-list and object-property
order are immaterial; iteration and effect order are not. Even an extra truthful
but unused dependency is refused as non-exact, not classified as malicious.

Only after this check does the store retain the original JSON string behind a
process-private handle. Later validation uses that retained string, not a
caller-edited copy. Cloned, serialized, other-store or forgotten handles are not
applicable. A failed state comparison need not destroy a handle: restoring its
observations can restore applicability. Forgetting one handle neither revokes
other checked copies nor prevents a new ingestion. These are evidence-lifetime
rules, not authority revocation or durable occurrence consumption.

\paragraph{Conditional checked-ingestion contract.}
Assume the accepted deterministic finite profile, the same workflow and actual
interpreter semantics and limits, honest immutable snapshots $x,y$ accepted by
input validation, and correct reference reconstruction, shared parser/helpers,
comparison and private storage. Let $e(C)$ be the ordered effects in candidate
$C$, and $\eta$ a handle issued by successful ingestion at $x$. Under the
workflow-identity and digest assumptions of Equation~\ref{eq:stability},
\begin{gather}
 \operatorname{Check}_{I}(C,w,x)=\eta
 \ \land\ \operatorname{Applicable}_{I}(\eta,w,y)
 \notag\\
 \Longrightarrow\quad F^q_y(w)=F^q_x(w)=e(C).
 \label{eq:checked-reuse}
\end{gather}
This is a conditional exact-trace statement about supported executions, not a
new algorithmic theorem or a proof of arbitrary-program closure.

\paragraph{Argument and trust boundary.}
Successful ingestion equates the candidate with the complete reference map
$Q_x(w)$ and trace $F^q_x(w)$. Applicability establishes that every retained query
has the same observation in $y$. Induction along the selected execution then
preserves literal and queried operands, branch choices, ordered iteration
scopes, and positive or negative lookup choices. An absent lookup retains its
outer scope; a present lookup binds the same item. The same argument recurses
through sequences and queued children, yielding the same ordered trace and
bounded expansion. Unvisited branch operands need not be dependencies when the
selecting predicate is unchanged, but all syntax must remain in the accepted
profile. This discharges the completeness premise only relative to the trusted
reference semantics and inputs. Shared parsing and helpers can share defects;
the interpreter identifier is a supplied label, not an attestation of executing
code. The handle and diagnostic hashes are neither transferable proofs nor
signatures. A compromised checker process or dishonest snapshot defeats these
assumptions. Trace equality does not discharge formation authority, signature
verification, revocation, quota, replay prevention or the race to an actual effect.

\subsection{Developer conformance and its denominator}
The semantic suite passed 69 tests: the first query-aware slice's 35 tests,
seven reference-traversal tests and 27 checked-store tests. Its broken-trust
counterexample remains visible. Checked ingestion refused 12 single-dependency
deletions from one authored compound workflow and six effect-list mutations:
empty, removed, inserted, reordered, changed destination and changed byte count.
These are conformance cases within the suite, not 18 independent attacks.

A finite product examines 256 ordered before/after pairs from 16 states varying
a branch, byte count, callback presence and irrelevant metadata. Forty-eight
pairs are applicable and preserve the reference trace; 208 are exact-query
non-reuses, not authorization denials or demonstrations of harmful effects.
Expected reuse is fixed by the authored semantics rather than inferred from the
checker response, but the evidence path and interpreter share code and input
checks. These are internal developer-conformance results, not independent
semantic labels, and the repeated product is not added to the first semantic
slice as another independent sample. No business side effect or timing benchmark
is part of this semantic suite. Full reference traversal, parsing, comparison
and retention are paid at ingestion; no net cost benefit for this layer has been
measured.

\section{Real SaaS Dispatch and Outcomes}
\label{sec:saas-pilot}
\subsection{Bounded composition with a remote provider}
A frozen pilot applies the query-aware kernel and checked-summary store to
actual GitHub issue comments. Ten authored schedules each run once with fresh
preparation and once with checked reuse, giving 20 cells. Each original work
item proposes two ordered comments of exactly 512 UTF-8 bytes, one for each of
two owner-authorized issue destinations. This is real third-party storage, not
the synthetic receiver of Section~\ref{sec:prototype}, but remains one controlled
workflow rather than customer traffic or independently curated evaluation.

The prototype imports the finite analyzer, separate interpreter, checked store
and the existing pure BAF matcher from
\path{packages/cava-core/src/execution-boundary.js}. A separate actor process
interprets the workflow and snapshot. The credential-bearing dispatcher checks
the submitted effects against the privately retained checked trace, which must
equal the signed lease's effect trace. Leases bind exact body hashes as well as
destinations, ordered effects and byte counts; equal-length text substitution is not authorized by
the symbolic byte-count semantics. Local Ed25519 issuance and the actual BAF
dependency form a bounded composition, not deployment of the complete
PCAA/CAVA/BAF/AREG/ZeroGate pipeline. These leases are not production PCAA
receipts, and the pilot is not a new AREG evidence-bridge measurement.

Both arms retain signature, original-grant, current-authority, revocation,
quota, order and occurrence checks. Administrative changes, reservation and
intent are recorded in a local SQLite journal; GitHub does not participate in
that transaction. Because the connector serializes and paces requests, the
final authority check runs in a synchronous callback after pacing, immediately
before transport dispatch. Each attempted dispatch also checks campaign halt
state. Moving this check before the connector's wait would leave another local
staleness interval. Under the trusted dispatcher and journal assumptions, the
check binds the outbound request to locally applicable permission at dispatch;
it does not establish permission at an atomic remote commit point or recall
an in-flight request.

\subsection{Observed outcomes}
Across the two arms, 24 actor submissions produced 12 admissions, ten denials
and two duplicate responses. Eight admitted work items completed both steps;
four were deliberately partial. The captured records contain 20 dispatched
POSTs and 20 confirmed comments with matching raw bodies. Table~\ref{tab:saas}
separates those effects from refused and reconciled work. There were 22 final
claim checks: 20 dispatches and two second-step refusals following revocation.
The two issue-creation setup effects are excluded from the governed-comment
denominator.

\begin{table}[ht]
\centering\small
\begin{tabularx}{\linewidth}{@{}>{\raggedright\arraybackslash}p{0.30\linewidth}>{\centering\arraybackslash}p{0.17\linewidth}Y@{}}
\toprule
Schedule & POSTs / confirmed comments & Work-item outcome across both arms \\
\midrule
Authorized & 4 / 4 & Two complete two-comment items \\
Irrelevant metadata churn & 4 / 4 & Two complete items; live authority checks retained \\
Revoked before release & 0 / 0 & Two refusals before dispatch \\
Runtime scope narrowed & 0 / 0 & Two whole-item refusals before dispatch \\
Absent callback added & 0 / 0 & Expanded consequences not admitted under old evidence/lease \\
Competing occurrences, one batch quota & 4 / 4 & One complete item per arm; the competing item refused \\
Duplicate delivery & 4 / 4 & One complete item per arm; two duplicate responses, no extra POST \\
Same-length body substitution & 0 / 0 & Hash mismatch refused before dispatch \\
First confirmation suppressed locally & 2 / 2 & One first effect per arm reconciled; no resend or second step \\
Revocation after first remote effect & 2 / 2 & First effects retained; second steps refused \\
\midrule
Total & 20 / 20 & Eight fully complete and four deliberately partial admitted items \\
\bottomrule
\end{tabularx}
\caption{Real GitHub pilot: ten schedules $\times$ two preparation arms, one cell per schedule/arm. Counts combine both arms, whose scheduled outcomes agree. Confirmation includes provider readback after local confirmation loss; denials and duplicate responses are not extra effects.}
\label{tab:saas}
\end{table}

In the confirmation-loss schedule, GitHub accepted the first write but the
prototype suppressed its local confirmation. The result was initially unknown;
provider readback reconciled one matching comment per arm without resubmission
or execution of the second step. Post-effect revocation likewise retained the
first comment and refused the second. No comment was removed to manufacture
atomicity, and these four partial items were not rolled back. The faults were
locally injected around real provider requests, not observations of a GitHub
outage or malfunction.

\subsection{Counted work, verification and inference limits}
Each arm used ten initial primes and eleven preparations; the extra preparation
was for the second quota contender. Duplicate delivery reused an already
prepared request. The fresh arm performed 21 full analysis/check operations.
The checked-reuse arm performed ten full operations and eleven reused
preparations, avoiding eleven repeated analyses without omitting live dispatch
checks. This is counted work, not elapsed time. The workload was intentionally
applicable at preparation, with relevant changes injected afterward; neither a
production hit rate nor a latency ratio follows.

The fresh arm's ten primes were unused because each preparation reanalyzes.
Thus eleven avoided analyses is the difference under this deliberately primed
protocol, not net savings against an efficient unprimed fresh baseline. Removing
those unused primes would leave eleven fresh analyses versus ten reuse-arm
primes plus eleven reuse validations; this is accounting, not an additional
measured experiment. No positive net-cost conclusion follows without charging
the respective validation and lifetime costs.

Offline replay accepted the recorded dispatches and confirmed comments, and
complete paginated issue lists contained no extra or duplicate campaign effect.
A separate post-run readback matched all 20 comments by identity, destination,
author and body. The verifier does not import the dispatcher decision function
or BAF matcher; it separately reconstructs the authored trace, footprint,
journal-order authorization and reservation accounting and checks signatures
and provider records. This is implementation-independent checking within the
same project, not external human labeling or a third-party audit. Provider
readback corroborates stored objects, not historical local policy state.

The experiment establishes neither an all-or-nothing remote batch nor global
exactly-once execution. A marker is not a provider-enforced idempotency key, and
absence in a bounded readback window is not permission to resend. Two issues,
one workflow and one repetition per schedule do not establish generalization
across services. Host integrity, issuer custody, administration and snapshot
capture remain trusted; the account's privileges exceed the connector allowlist.
The queued-child node remains a trace construct. The confirmation-loss schedule
is not a process-restart experiment, and process-private handles and in-memory
signing setup do not constitute a restart protocol. The demonstrated increment
is checked consequence reuse composed with live local authorization and honest
non-atomic outcome handling, not distributed recovery or a full product stack.

\section{Durable Evidence and Conservative Recovery}
\label{sec:durable}
\subsection{A persisted summary is not a recovered authorization}
The query-aware store's private handle is useful within one process but cannot
be serialized into permission after that process disappears. The durable
realization in \path{packages/matoi-runtime/} retains the original workflow,
issuance snapshot, summary bytes, checker version and implementation identity.
Reopening runs checked ingestion again; a saved digest alone is insufficient.
It reconstructs a process-private handle only after the complete footprint and
effect trace pass the same finite checks. Current-state applicability is then
tested separately from authority. The two promoted semantic files are unchanged
from the earlier query-aware implementation; the new contribution is their
explicit lifecycle contract, not a retrospectively enlarged semantic language.

An Ed25519 issuer supplied by the trusted embedding service signs each prepared
occurrence. The envelope binds scope, issuer and implementation, the existing
action fingerprint, actor and session, workflow and summary, ordered effects
and exact body hashes, the original boundary, authority epoch and expiry.
Admission checks the whole envelope and reserves the trace's byte budget in a
SQLite transaction: available budget is debited at admission. Claim validates
the retained reservation and pending step, ordered predecessor completion,
current evidence, body binding and live authority; it is not a fresh
available-budget test and does not debit the reservation again. Identical
duplicate requests return the recorded step
states rather than acquiring another reservation. The journal does not store
the private key; an authorized replacement must obtain the same issuer from
its embedding service. Scope, source-profile or issuer mismatch refuses recovery
rather than silently resetting history.

Authority updates increment an epoch, including reactivation. This is distinct
from a state value changing and later returning to its original value: the latter
can restore evidence applicability, but cannot restore an earlier authorization
epoch. A new grant or occurrence is not inferred from a restart. The actual
embedding service, administrator identity and key custody remain trusted and
are not supplied by a serialized agent request.

\subsection{One local owner, three crash meanings}
The journal serializes state changes and their events using
\texttt{BEGIN IMMEDIATE}, WAL and \texttt{synchronous=FULL}. Explicit recovery
compares the expected owner generation and advances it in the same transaction
that restores evidence and classifies interrupted steps. Every subsequent
mutation checks that generation. This fences an old handle from new journal
operations; it does not recall transport already claimed or sent. State and
event export use one read snapshot, so a concurrent update cannot produce a
report combining earlier state with later history.

There are three importantly different interruption points. Before an intent,
a reserved step is pending and may later pass fresh dispatch checks. After an
intent but before an observed result, it is unknown: the request may have left,
even when the experiment deliberately stops it before HTTP. After a receiver
commit but before local confirmation, that same unknown state may correspond
to a real effect. The recovering process must not distinguish these last two
by guessing from the absence of a local success record.

Recovery therefore converts intent to unknown and cancels the affected
unstarted suffix. Cancellation refunds only that suffix. The uncertain prefix
continues to consume its reservation. A trusted provider readback must match
the exact destination, body, principal and receiver-recorded dispatch-attempt
correlation. Matching body content without the correlation is insufficient.
For a nonconflicted unknown step, one matching observation records a historical
confirmation, even if authority has since been revoked; it does not resume the
cancelled suffix. Zero matches preserve unknown. Multiple distinct matching
records supplied in a single reconciliation batch while the step is unknown
mark a conflict rather than choosing a convenient success. A conflict-marked
unknown step remains unknown even if a later readback
supplies only one exact match; neither confirmation nor reconciliation clears
that conflict, including after recovery. Its reservation remains charged.
Already confirmed steps are not re-audited for later-discovered duplicates or
contradictions; no late-conflict-based blocking of dependent claims is provided.
This is the approved single-batch reconciliation scope, not a later audit
service. The adapter must obtain the correlation
from receiver evidence, not attach the local attempt to an unrelated old object.

\begin{proposition}[Conservative local recovery]
For the supported profile, assume one authoritative, protected, non-rolled-back
and non-forked journal history, trusted implementation and issuer, atomic journal
transactions, and all dispatch attempts using the
claim transition. For each admitted occurrence and step, at most one claim
transitions pending to intent, including across owner recovery. Fix a grant $g$,
with original budget $Q_g$ and available budget $q_g$. Let $\mathcal{S}_g$ contain
the steps of all admitted occurrences rooted in $g$, and let $c(s)$ be step cost.
The journal preserves, independently for each grant,
\[
 q_g + \sum_{\substack{s\in\mathcal{S}_g:\,\operatorname{state}(s)\in\\
 \{\mathrm{pending,intent,unknown,confirmed}\}}} c(s) = Q_g.
\]
Conflict-marked unknown steps remain in the sum. This is a local reservation
and dispatch-claim property, not a bound on effects a provider may generate.
\end{proposition}
\begin{proof}
Registration starts with $q_g=Q_g$. Atomic admission subtracts exactly the costs
of newly pending steps; identical duplicates subtract nothing. Claim requires
pending and changes it to intent without a second available-budget debit.
Confirmation and reconciliation can change a nonconflicted intent or unknown
to confirmed without altering cost; a conflict-marked unknown cannot make that
transition. Stop and recovery change intent to unknown, and move only pending
steps to cancelled while refunding
their costs once. No transition restores intent, unknown, confirmed or
cancelled to pending. The generation check excludes later mutations by an old
owner. Each transition preserves the equation and none supplies a second claim
for the same step. Remote execution is outside that serialization and is not
part of the conclusion.
\end{proof}

Restoring an older valid backup or operating copied journals independently
violates the single-history premise. Matching issuer and source identity does
not prove that a database contains the latest consumption history. This package
has no external anti-rollback or fork-detection service.

\subsection{Dispatch is still not remote commitment}
The dispatcher waits for provider pacing before its final synchronous claim,
then invokes the adapter without an intervening await. A send or readback error
leaves uncertainty; there is no automatic mutation retry. The adapter may still
queue internally, and a takeover or revocation may occur after the claim.
Neither a generation token nor SQLite makes that interval atomic with a remote
provider. This realization consequently offers local claim fencing and
conservative observation-based reconciliation, not remote fencing, rollback,
global exactly-once delivery or an all-or-nothing workflow. An unknown result
may remain charged indefinitely; resolving that availability cost needs an
explicit provider/operator contract, not a timeout that grants permission again.

\subsection{Real local process recovery}
The recovery protocol was frozen before implementation and the retained run.
Ten authored schedules each run once in fresh and checked-reuse arms, producing
20 cells with the same authority and history checks. A finite two-publication
workflow has bodies of five and six bytes under an eleven-byte grant. Actual
owner processes and a loopback HTTP receiver use separate SQLite databases;
the controller injects SIGKILL at recorded IPC barriers. There are 32 owner
starts, 20 receiver starts and ten owner kills. This is not external-SaaS restart
recovery, machine power loss or a latency experiment. Quota competition is two
submissions to one active owner, not multiple independent journal authorities.

\begin{table}[ht]
\centering\small
\begin{tabularx}{\linewidth}{@{}>{\raggedright\arraybackslash}p{0.30\linewidth}>{\centering\arraybackslash}p{0.14\linewidth}Y@{}}
\toprule
Schedule & POSTs / effects & Final outcome across both arms \\
\midrule
Ordinary completion & 4 / 4 & Two complete occurrences \\
Pending crash and recovery & 4 / 4 & Fresh checks precede resumed dispatch \\
Intent before HTTP, then crash & 0 / 0 & Two unknown prefixes; suffixes cancelled \\
Receiver commit before confirmation, then crash & 2 / 2 & Two prefixes reconciled; no resend or suffix \\
Revoke after first effect & 2 / 2 & Two second-step authority refusals \\
Duplicate admission & 4 / 4 & Two duplicates; no additional effect \\
Shared quota, two candidates & 4 / 4 & One complete candidate per arm; rival denied \\
Wrong issuer recovery & 0 / 0 & Two refusals; four original steps stay pending \\
Stale owner during pacing & 4 / 4 & Two stale claims fenced; replacements complete \\
Ambiguous receiver readback & 2 / 4 & Two injected provider duplications; conflicts retained \\
\midrule
Total & 26 / 28 & Ten complete cells; four unknown-containing cells \\
\bottomrule
\end{tabularx}
\caption{Local recovery campaign, one repetition per schedule/arm. Effects are actual receiver rows, not all confirmed steps. The duplication fixture deliberately violates the normal one-effect-per-POST provider assumption; detecting its ambiguity is not preventing its duplicate effects.}
\label{tab:recovery}
\end{table}

All 20 cells match the frozen oracle, but only ten contain a complete admitted
occurrence; six meet the stricter clean-completion indicator, which excludes
quota refusals and stale-owner fences within the same cell. Twenty-four
submissions yield 20 admissions, two duplicates and two quota denials.
The 40 admitted steps end as 24 confirmed, four pending, eight cancelled and
four unknown. Two reconciliation conflicts are included in unknown, not extra
steps. The two injected duplicate-producing requests account for four receiver
rows that are not confirmed as successful steps. The intent-before-HTTP cases
retain uncertainty despite having no receiver effect: the recovery path does
not infer that fact from a missing local confirmation.

The independent receiver databases contain 26 POSTs and 28 rows. There are 48
dispatcher calls including no-resend probes, but 28 adapter send entries and
26 actual POSTs; two owners were killed inside the adapter before HTTP. Unknown
prefixes remain charged five bytes and only unstarted six-byte suffixes are
refunded. The fixed expected states, observations, process exits and per-grant
conservation were checked without importing the runtime decision function.

Post-run replay reopens temporary copies of all 40 databases, checks 31 captured
source files and 200 cell artifacts, and verifies 62 signed request records:
42 preparations including primes/contenders, plus 20 admitted copies. Public
SPKI keys are retained, not private keys. The final package's 93 passing tests
include signature/archive corruption negatives and temporary process cases;
they are not added to the retained 20-cell denominator. These are same-project
engineering checks, not independent labels or attestation of the originating
host. The original process and GitHub archives remain separate and unchanged.

\section{Same-Kernel Qualification Against Protected Caches}
\label{sec:qualification}
\subsection{Match the security obligations, not just the workload}
The preceding studies leave a substantial comparator question: does checked
reuse offer an advantage over well-protected evidence produced directly by a
trusted reference interpreter? A successor experiment adds three closed
preparation strategies to the same durable kernel. Reference-fresh derives
evidence for every new occurrence. Reference-snapshot reuses it only after
whole-snapshot equality. Reference-cache reuses it when its exact query
observations remain applicable. Checked-reuse performs full finite derivation
checking on creation or restoration and query validation on reuse. The older
fresh mode, which includes analysis and checking, remains compatible but is
not relabelled as a single reference traversal.

All four primary arms retain the same signature, issuer/profile, original
grant, current authority/epoch, expiry, coupled effect, quota, duplicate,
generation-fencing and final-claim checks. They use the same durable state
machine and provider adapter. Snapshot equality only selects preparation;
it cannot substitute for current authorization. Restored reference records
are checked against their retained issuance inputs under protected-journal
custody, not advertised as independently validated untrusted derivations.
The baseline's store is not made attacker-writable to manufacture a treatment
advantage. These are local implementations, not measurements of FCD,
AgentSpec or commercial governance services.

\subsection{Full-lifecycle cost with no unused primes}
The confirmatory cost protocol fixes twelve authored cases, four arms and
twenty paired fresh-process blocks, yielding 960 cells. It varies occurrence
reuse counts of one, eight and 64, effect counts of one, sixteen and 64,
small/large unread state, used-value changes, branches, ordered collections
and negative lookups. This is a selected response surface, not a complete
Cartesian product or independent real-tool corpus. Every new occurrence has
a distinct identity; duplicate probes are not counted as reuse. Arm order
is seed-randomized within a block, and there is no warmup or free prime.

Protocol and imported source hashes are retained and printed before the
first cell. An earlier v1 collection intended as confirmatory was stopped
after 82 complete cells when review identified insufficient evidence-oracle
checks. That partial archive and all pilots are retained but excluded. The
v2 protocol was frozen before the successor collection, not selected after
inspecting its performance. This is a local pre-data freeze, not external
preregistration or independent annotation.

The host is an Apple M4 Pro with fourteen logical CPUs and 24 GiB memory,
Darwin 27.2.0/arm64 and Node 24.19.0. Other project experiment workloads were
paused; unrelated workstation/OS activity is not claimed to be eliminated.
Each cell starts real owner and receiver processes, fresh SQLite stores and
an ephemeral Ed25519 issuer. Lifecycle time includes key creation, startup,
preparation, signing, persistence, live checks, HTTP, readback, metrics and
cleanup. Artifact extraction/write/hash and replay have separate boundaries.
The complete campaign takes 875.204 seconds including per-cell verification,
excluding final manifest sealing; that aggregate is not service latency.

All 960 cells complete equal prescribed work: 16,080 occurrences and 162,640
actual receiver mutations. Table~\ref{tab:protected-cache-cost} reports every
prespecified comparison with the strongest protected-cache baseline. Ratios
use paired block sums and a seeded 2,000-draw percentile bootstrap. Intervals
are descriptive, pointwise within-host intervals, not simultaneous
multiple-comparison guarantees, equivalence tests or production percentiles.

\begin{table}[ht]
\centering\small
\begin{tabularx}{\linewidth}{@{}Yrr@{}}
\toprule
Case & Checked / protected cache & 95\% interval \\
\midrule
Stable, cold one-occurrence & 1.012 & [1.004, 1.020] \\
Stable, eight occurrences & 1.003 & [0.991, 1.015] \\
Stable, 64 occurrences & 1.004 & [0.995, 1.014] \\
Unrelated drift, sixteen effects & 0.991 & [0.965, 1.011] \\
Unrelated drift, long reuse & 0.999 & [0.993, 1.004] \\
Relevant drift, sixteen effects & 0.980 & [0.933, 1.010] \\
Stable, large state & 1.000 & [0.995, 1.005] \\
Unrelated drift, large state & 1.003 & [0.992, 1.014] \\
Wide, 64 effects & 0.997 & [0.991, 1.003] \\
Branch flip & 1.000 & [0.996, 1.005] \\
Membership reorder & 1.002 & [0.991, 1.012] \\
Negative-lookup insertion & 0.934 & [0.821, 1.008] \\
\bottomrule
\end{tabularx}
\caption{Full cell lifecycle cost at equal completed work, twenty paired blocks
per row. Below one favors checked reuse. Eleven intervals include one; the
remaining cold-start interval is more expensive. The 0.934 point estimate is
not evidence of a demonstrated 6.6\% speedup.}
\label{tab:protected-cache-cost}
\end{table}

This result does not establish a full-lifecycle performance advantage for
checked reuse over a competent protected cache. It also does not show that
checking and protected provenance have identical trust requirements. The
separate ingestion question concerns whether supplied derivations are valid;
the trusted baseline may safely ignore them and derive its own evidence.
Under the primary honest-store model, neither strategy is made deliberately
unsafe. The whole-state JSON journal, signatures and live execution checks
remain in every arm. Fewer analysis calls alone do not eliminate those costs;
this measurement does not separately identify per-fsync cost or prove a
particular persistence optimization would improve them.

\subsection{Fixed offered load, including refused work}
A separate protocol offers 64 new occurrences per cell at 25, 100 or 400
per second, with at most eight in-flight harness requests. Twenty blocks
per arrival rate and arm yield 240 cells and 15,360 offered occurrences.
There are 14,819 completed occurrences, 541 harness-capacity refusals and no
remaining unknowns. At 25 and 100 per second, every arm completes all 1,280
offered occurrences at each rate. Table~\ref{tab:qualified-load} retains the
rejections at the highest offered rate rather than timing only successes.

\begin{table}[ht]
\centering\small
\begin{tabularx}{\linewidth}{@{}Yrrr@{}}
\toprule
Arm at 400 offered/s & Offered & Completed & Capacity refusals \\
\midrule
Reference-fresh & 1,280 & 1,142 & 138 \\
Reference-snapshot & 1,280 & 1,152 & 128 \\
Reference-cache & 1,280 & 1,142 & 138 \\
Checked-reuse & 1,280 & 1,143 & 137 \\
\bottomrule
\end{tabularx}
\caption{Twenty short batches per arm at the highest frozen offered load.
The eight-slot limit belongs to the harness, not a certified production queue.
The lower rates complete all offered work. Refused work stays in the denominator;
these unequal-completion cells do not produce an equal-work latency ranking.}
\label{tab:qualified-load}
\end{table}

Arrival plans, actual arrival lag, returned times and receiver rows are
retained separately. These short fresh-journal batches do not measure a
sustained long-lived service, multi-host scaling, tenant scheduling or a
production SLA. Their useful-completion counts do not establish that a
400/s client offered rate is sustainable production throughput, or that
cache choice removes persistence and transport costs.

\subsection{Retention growth is charged, not primed away}
Sixty additional cells use five preparation counts (255, 256, 257, 900 and
1,025), four arms and three fresh-process repetitions. The first three
straddle the 256-entry loaded-evidence limit; 1,025 deliberately exceeds
the 1,024 retained-preparation cap. Body size varies by the fixed input rule
to create distinct dependency evidence. This is not a one-variable causal
study of journal length with payload size held constant.

Across the matrix, 32,316 preparation attempts retain 32,304 preparations
and produce twelve explicit capacity refusals. Each cell then restores the
first input and executes its previously prepared occurrence: all sixty
probes complete. These are sixty actual effects, not 32,304 executed actions.
Checked-reuse median complete fill-and-probe cost grows from 664.838 ms at
255 attempts to 7,285.965 ms at 900 and 9,463.619 ms at 1,025. The corresponding
reference-cache medians are 659.038, 7,177.978 and 9,440.384 ms. Three
fresh-process samples per preparation-count/arm combination are insufficient
for a production tail-latency estimate.

The explicit cap preserves a bounded implementation but is not a long-lived
retention solution. The journal still serializes whole state and has no
compaction or authorized release procedure for unknown reservations. A
deployment must not recycle occurrence history or uncertain quota merely to
avoid that limit. The measured preparation and body-size growth motivate
profiling persistence separately; they do not establish that any proposed
incremental storage design is already correct or faster.

\subsection{Fault outcomes remain in the result}
The qualification also repeats thirteen frozen fault schedules across all
four strategies in three fresh-process blocks: 156 cells and 180 perform
calls, including explicit duplicate and quota contenders. They produce 156
admissions, 168 POSTs and 180 receiver rows. Thirty-six actual SIGKILL
injections distinguish loss before intent, after intent and after receiver
commit. Only 48 occurrences complete in full. The 312 admitted steps end
as 144 confirmed, 132 cancelled and 36 unknown. Twelve duplicate-observation
conflicts are included within unknown, not added as extra successful effects.
Unknown steps retain 180 bytes of reserved quota.

Revocation and narrowing during pacing each prevent all sends in their
twelve cells. In contrast, revocation after claim permits the already
claimed first send in twelve cells and refuses the second effect. This
records the local-claim/remote-commit gap rather than claiming it away.
Twelve intent-before-POST crashes retain unknown prefixes despite zero
receiver rows; the runtime does not infer non-execution from an absent
observation. Exact later observations reconcile twelve after-commit crashes,
twelve suppressed-readback cases and twelve wrong-principal cases, without
restarting cancelled suffixes. The wrong initial principal is rejected.
Twelve deliberately duplicate-producing requests create 24 rows and remain
conflicts. Expiry while unresolved does not refund unknown reservations.
Duplicate admission, quota contention and stale-owner schedules each complete
one occurrence per cell without additional uncontrolled sends.

Readback suppression, principal substitution and clock advancement are
explicitly injected local faults. A one-second observation horizon is
fixed before collection; no new mutation is sent to resolve uncertainty.
Killed perform calls lack returned stage timings, so fault stage aggregates
are not zero-cost measurements or latency rankings. Their lifecycle and
controlled-barrier CPU observations remain retained. All 1,416 qualification
cells pass the captured-source replay with pinned manifest digests,
including separate database reads and exact signed-request reconstruction.
This is reproducibility under an honest recorded host, not attestation of
that host or independent customer-workflow ground truth.

\section{Testing the Complete-Mediation Premise}
\label{sec:qualified-coverage}
Matoi's commitment argument assumes an effective control on every covered
consequence path. A supporting experiment, called Defensible Runtime Space
(DRS) in the repository, tests that premise rather than adding another
evidence-reuse algorithm. A guard listed in a graph can be ineligible to
serialize a mutation, fail activation, or miss a physical route. Selection,
qualification and actual enforcement are therefore recorded separately.

\subsection{A fair supplied-inventory comparison}
A gateway, queued workers and effect owners run as real processes. The
concrete effects are document-revision insertion, document-access grant and
notification enqueue in SQLite WAL/FULL databases. Enqueue is not email
delivery. At a qualified owner, current authority, budget consumption,
occurrence deduplication, the business mutation and its receipt commit in
one shared transaction domain. This is a separate realization of a mediated
owner, not a deployment of the durable Matoi package or a distributed SaaS
transaction.

Owner-only enables each reachable qualified commit owner in the visible
inventory. All-exposed adds working gateway/worker checks but retains the
same final owner checks; early denial really stops downstream work.
Path-selected uses the existing finite minimum-cover selector. All arms
see the same owner/capability inventory, credentials and activation evidence.
None receives the evaluation oracle's hidden physical paths or a repair
unavailable to its comparators.

The frozen protocol has eight configurations, five schedules, three arms
and three repetitions, giving 360 cells. Schedules cover authorized fan-out,
invalid binding, queued revocation, one-unit budget contention and concurrent
duplicate delivery. A separate oracle enumerates the supplied physical
routes and subsets of at most eight guard sites. Another reconstructs
authorization, exact-effect binding, serialized budget/occurrence history
and actual business rows. They do not import the selector or runtime
decision helpers, but are same-project implementations with authored
expectations, not external semantic labels. The host/controller and physical
manifest remain trusted.

\begin{table}[ht]
\centering\small
\begin{tabularx}{\linewidth}{@{}Yrrrr@{}}
\toprule
Configuration & Cells & Qualified & Refused & Unauthorized effects \\
\midrule
Independent owners & 45 & 45 & 0 & 0 \\
Shared atomic committer & 45 & 45 & 0 & 0 \\
Diamond routing and retry entry & 45 & 45 & 0 & 0 \\
Missing atomic hook & 45 & 0 & 45 & 33 \\
Activation failure & 45 & 0 & 45 & 33 \\
Declared raw writer & 45 & 0 & 45 & 45 \\
Explicitly incomplete inventory & 45 & 0 & 45 & 69 \\
Falsely complete inventory & 45 & 0 & 0 & 45 \\
\bottomrule
\end{tabularx}
\caption{Physical-oracle coverage and actual recorded effects, not just
selector output. Refused configurations are still deliberately probed in
diagnostic mode; their effects are not qualified deployment work. The
falsely complete inventory incorrectly claims coverage in all 45 cells.
The 45 explicitly unknown deployments are included in the 180 refusals.}
\label{tab:qualified-coverage}
\end{table}

\subsection{Coverage failures are retained, not scored as protection}
All 360 cells produce complete replayable evidence: 1,071 recorded requests
and 585 business mutations. In the 135 physically qualified cells, there
are 135 distinct authorized receipted \emph{per-effect} occurrences, 189
denials and zero observed unauthorized effects or false denials. The
coincidental equality of cell and effect counts does not mean 135 entire
workflows completed. These five authored schedules do not estimate a
real-world false-allow probability.

The remaining results are materially different. Four configurations
refuse deployment in 180 cells, including 45 explicitly unknown cases.
The controller deliberately runs diagnostic probes after refusal to expose
what would happen if it were ignored; these produce 180 unauthorized
mutations. In the falsely complete configuration, a hidden raw writer
produces another 45 unauthorized mutations alongside 45 legal per-effect
occurrences. All three strategies accept the same misleading inventory.
Thus one configuration, repeated over 45 cells, is falsely labelled covered;
neither path selection nor owner-only instrumentation discovers the omission.

Across unqualified runs, 294 recorded mutations lack atomic occurrence/budget
consumption. That is a different denominator from the 225 unauthorized
mutations: a raw write can happen to satisfy the current effect bound while
still lacking a serialized enforcement record. The recording-only audit
path establishes that a write occurred, not that it was controlled.

The 225 unauthorized mutations span five unqualified configurations:
missing atomic hook, failed activation, declared raw writer, explicitly
incomplete inventory and falsely complete inventory, respectively
$33+33+45+69+45=225$. They are not all inventory-omission failures. In
particular, the false-complete case has 45 false-covered cells and 45
unauthorized writes; deployment refusal in the other configurations is
not an enforced stop on the deliberately continued diagnostic probes.

\subsection{More controls are not a free baseline}
Within the qualified subset, each arm has 45 cells and 45 useful per-effect
completions. Owner-only and path-selected both use 117 actual checks, 351
RPCs and 990 complete-lifecycle transactions; both have zero placement
regret against the finite unit-cost oracle. All-exposed performs 291 checks
but only 300 RPCs and 939 transactions because early refusals save downstream
work. Its summed placement regret of 90 counts two additional unit-cost
sites per qualified cell, not 90 milliseconds or a measured service loss.
Counting only guards would miss that operational trade-off.

The full matrix charges 7,914 explicit commits: 4,680 schema, 360 authority
initialization and 2,874 work commits. Schema readback and WAL commit frames
validate initialization accounting. Planning consumes 16.078 seconds of
aggregate wall time across all cells and 12.512 seconds of selector CPU,
recorded separately from Node child CPU. Process/barrier-inclusive execution
times are not a production latency ranking. No placement-speed superiority,
new set-cover algorithm, malicious-host isolation or automatic topology
discovery follows from this experiment. Its contribution to Matoi is a
tested boundary: a valid evidence/authority check still needs a qualified,
active and unbypassable effect owner.

\section{Portable Evidence and Explicit Durable Integration}
\label{sec:portable-update}
This section updates the research state without replacing the preceding
experiments. It concerns a fixed bounded publication language, not the
original refund interpreter, arbitrary tool source code or the earlier
GitHub dispatcher. The detailed finite-contract manuscript and supplement
accompany this long-form preprint as ancillary documents. All three documents
attribute each measurement to its actual implementation and campaign.

\subsection{An exact result that another implementation can check}
The finite result contains an ordered publication-effect trace and the exact
active query-observation map. A publication names a destination and positive
byte count, not exact body bytes; it does not determine or classify the message
body. Durable host body hashes bind the actual bytes separately, and the
Banking adapter's numeric CAVA constraints are a third, distinct check
(Section~\ref{sec:banking}). Sequence,
branch, trace-only queue, ordered iteration and present/absent lookup complete
the language. Absence, collection order and lexical binding are observable
dependencies, not optional producer annotations. Unknown syntax or missing
required input causes refusal rather than an empty successful trace.

An accepted repair must equal fresh reference evaluation in both effects and
active observations. Reused subtrees come only from the consumer's privately
retained, previously checked result. Replacements are checked against current
syntax, snapshot and binding context. The conditional argument is a finite
correctness result, not a novel general incremental algorithm.

\begin{proposition}[Fixed-context finite subtree reuse]
For the accepted deterministic v2 profile, fix the workflow and interpreter
identity, a fixed syntax position, and the full binding context (the complete
ordered lexical binding stack). Assume a correct checker and fresh semantics,
collision-resistant identity checks and honest immutable snapshots. A
previously accepted subtree whose complete active observations still match
the current snapshot has the same local ordered effects and observations
as fresh evaluation at that position and binding context.
\end{proposition}
\begin{proof}
At the fixed position, literal operands and syntax agree; matched observations
preserve queried operands, branch choice, ordered membership and positive or
negative lookup. The full binding context preserves item selection, including
outer scope when a lookup is absent. Induction over the selected finite
execution preserves each child and the concatenated local effect trace.
This argument does not permit moving a record to another syntax position or
binding merely because a query value or destination happens to agree.
\end{proof}

Whole-workflow composition still enforces global visit, effect and distinct-query
bounds from the combined records, validates current input and result sizes,
and freshly assigns global effect ordinals over the assembled ordered trace.
Subtree records do not supply reusable absolute ordinals. The local lemma
alone cannot discharge enclosing resource bounds or authorize dispatch.

The explicit profile \texttt{matoi-finite-evidence-v2} fixes canonical UTF-8
bytes, UTF-16 key ordering, safe integer/scalar rules, resource bounds and
refusal precedence. Canonically encoded workflow and state are each bounded by
262,144 bytes, not by one shared 262,144-byte allowance. The whole wire request
is bounded by 9,437,184 bytes excluding framing LF, and the canonical proposal
by 8,388,608 bytes. These are separate bounds with their own depth/value limits.
This is a restricted application of the existing JSON
Canonicalization Scheme~\cite{jcs}, not new cryptography. A historical probe
found that the earlier prototype's host-locale sorting gave different
checkpoint identities under \texttt{en-US} and \texttt{sv-SE} for identical
non-ASCII keys. The new profile does not rewrite those old identities or
invalidate the old within-host observations; it separately defines the
portable contract and tests it. V2 is not a retroactive fix to the historical
locale-dependent implementation or its frozen measurements.

\paragraph{Source, checkpoint and record identities.}
Raw workflow source $s$ and portable workflow bytes are different objects.
The durable wrapper parses $s$ and canonically encodes the resulting bounded
JSON value before portable checking; the strict portable wire itself is not
arbitrarily reformatted. The domain-separated digest of those canonical
workflow bytes, together with interpreter/profile identity, identifies the
finite program. A full \texttt{checkpointId} additionally binds the canonical
accepted state digest and exact result. Valid source reformatting can preserve
these portable identities, whereas the signed request's \texttt{workflowHash}
commits to raw UTF-8 source and changes with its spelling.

An unread-state change can preserve retained evidence applicability and result,
but a new checkpoint checked at that state has a different state-bound ID.
The durable record's \texttt{summaryHash} is different again: it binds raw
issuance workflow/snapshot strings, the materialized full proof, checkpoint
and checker/implementation/profile metadata. Digest knowledge and serialized
handles confer no private checked custody. Rechecking complete evidence on
restart does not itself establish current applicability or authority.

JS and Python consumers were authored from the common specification and
vectors without sharing implementation source. This reduces some correlated
implementation errors but remains same-project AI-assisted work. Within each
language, the checker and fresh evaluator share parsing, input validation,
query observation and local-rule helpers. Agreement within one language is
therefore not an independent check of those shared semantics; the
opposite-language comparisons also remain same-project evidence. A complete
trusted snapshot is still supplied by the host. Neither a canonical hash nor
an interpreter label authenticates its source or proves that it describes
every real external effect.

\subsection{Negative repair results and a narrower positive region}
The first repair campaign compares five strategies over 1,200 process cells
and 18,800 evaluations, repeating 188 authored state positions. All complete
results and exact query footprints match the older same-project oracle.
Against matched full checking, aggregate preparation rises 3.24\% and child
lifecycle rises 1.06\%. Fewer new rules are checked, but full input, proof and
output costs remain. These negative observations are retained.

The portable campaign changes the transport boundary: repair may send compact
references instead of complete proof subtrees. Twelve families, three arms,
two consumer languages and twenty blocks yield 1,440 cells. Its 17,280 outputs
repeat 144 authored positions and match the opposite language's fresh result.
Cold process startup, producer, complete snapshots, IPC, output and shutdown
are charged. Archive writes and later replay are outside the cell clock.

\begin{table}[htbp]
\centering\small
\begin{tabularx}{\linewidth}{@{}Yrrr@{}}
\toprule
Consumer & Repair/full whole-cell ratio & 95\% interval & Reduction \\
\midrule
JS & 0.895956 & [0.889913, 0.902467] & 10.40\% \\
Python & 0.837778 & [0.832257, 0.842638] & 16.22\% \\
\bottomrule
\end{tabularx}
\caption{Within-profile portable cost contrasts. These are neither the older
repair experiment nor measurements of an integrated runtime.}
\label{tab:portable-update-cost}
\end{table}

In that historical campaign, high-change families retain null results,
protected fresh evaluation remains cheaper, and a matched protected-incremental
portable comparator was not run. Its combined protocol benefit does not
isolate semantic reuse from ordinary transport caching. The separately frozen
control in Section~\ref{sec:matched-cost} now supplies an identical-compact-input
comparison: one arm resolves accepted references but rechecks every rule,
while the other checks applicability and reuses accepted subtrees. The new
full-packet contrast also changes encoding, reconstruction, parsing and
transport, so it is not a pure transport-only estimate. Neither campaign
establishes superiority over a competent external protected-incremental
implementation, and no evidence-only percentage is imported into a customer
latency claim. The old table and its measured ratios remain unchanged.

\subsection{Evidence is connected to dispatch, not promoted into permission}
The explicit durable profile
\texttt{matoi.bounded-publication.portable-v2} supports
\texttt{portable-full} and \texttt{portable-repair}. It connects the portable
checker to the existing current-authority, signature, occurrence, quota and
recovery state machine. The host binds actual body hashes and exact request
order separately from symbolic destination/byte-count effects. CAVA action
identity is host-supplied; the profile does not infer arbitrary runtime meaning.

Before persistence, compact proofs are expanded into complete records. On
restart the consumer checks those records against their issuance inputs rather
than importing a serialized trusted handle. Admission and final claim still
check the current workflow/snapshot, original and current bounds, authority
epoch, expiry, signature, body and occurrence history. An applicable proof
cannot revive revoked authority, and a live grant cannot justify stale effects.
Available byte budget is debited at admission; final claim checks the retained
reservation, not fresh available budget for a second reservation. A
nonconflicted possible prior dispatch remains unknown until an exact trusted
readback confirms it. A conflict-marked unknown remains unknown despite a later
single match. Conflicts are evaluated together in a batch while unknown;
confirmed steps are not re-audited for later contradictions. Missing
confirmation cannot create an automatic mutation retry.

Two authored publication workflows, three strategies, seventeen schedules
and three repetitions yield 306 local process cells. Expectations were fixed
before collection. They are same-project symbolic cases, not outside human
labels. All cells and captured-source replay match those expectations.
Across the matrix, 702 child processes issue 252 POSTs and leave 270 receiver
rows and 2,700 journal events. Admitted steps finish as 234 confirmed,
162 cancelled and 36 unknown.

Seven zero-send refusal schedules account for 126 cells. They cover stale or
missing evidence, revoked/reactivated grants, attempted widening of the original
boundary and altered request bodies. Eighteen narrowing cells retain the first
effect and refuse the second. The 36 negative controls are not prevention wins:
eighteen expose one POST producing two receiver effects, and eighteen permit
an already claimed send to commit after revocation. The local serialization
point does not atomically control a remote provider.

\begin{table}[htbp]
\centering\small
\begin{tabularx}{\linewidth}{@{}Yrr@{}}
\toprule
Integrated strategy & 102 sessions, total ms & Ordinary-session median ms \\
\midrule
Portable full & 11600.239037 & 101.561542 \\
Portable repair & 11499.018373 & 102.770250 \\
Protected reference cache & 11155.197332 & 99.158792 \\
\bottomrule
\end{tabularx}
\caption{Complete local sessions include startup, two preparations, dispatch,
readback, diagnostics, prescribed recovery and shutdown. Six ordinary samples
per arm and fixed scheduling do not establish a speed ranking or tail SLA.}
\label{tab:portable-update-integrated}
\end{table}

Preparation repair checks 681 new rules instead of 1,410 and sends 363,414
encoded proof bytes instead of 505,440: reductions of 51.70\% and 28.10\%.
Yet recorded preparation-checker time is only 169.015671 ms in the full arm,
or 1.46\% of its mixed-session total, and 136.367503 ms in repair. Those partial
timers omit restoration and failed preparation; their ratios are descriptive,
not a causal decomposition. Both arms retain the same complete record-encoding
volume and perform sixty full checks during recovery. Optimizing the count of
new rules alone does not establish a useful end-to-end improvement.

\subsection{What remains outside this integration}
The contract assumes a trusted host, correct checker, authenticated provider
observations and one non-forked authoritative journal history. It does not
include tenant administration, operational key rotation, long-lived retention,
arbitrary providers or a public hosted deployment. The application image's
Node-20 runtime and the research kernel's Node-24 requirement need an explicit
integration boundary. No Hook is enabled and no customer action is controlled
by this study. DRS still assumes a supplied inventory; Shuwen's independent
skill-boundary research and reserved evaluation inputs are unchanged.

\section{Externally Defined Message Plans}\label{sec:agentdojo}
The next evaluation increment uses AgentDojo's original task definitions,
local tools and utility/security predicates~\cite{agentdojo}. Commit
\texttt{089ed468cf3ed0322acc66b0211f26d9d90dbf60} and its dependency lock are
pinned before collection. The original \texttt{v1.2.2} registry resolves
Slack suite 1.2.0 with 21 user tasks and five attack goals. Ten whole user
tasks fit the direct/channel-message profile. One task has no covered
mutation, nine contain unsupported invitation/membership changes, and one
publishes to a website. Those eleven remain in the coverage inventory.

This is a ground-truth-tool-plan replay, not an LLM-agent or generated
prompt-injection experiment. Literal summary placeholders are unchanged.
The original Python tools mutate their real in-memory environment; they do
not connect to Slack. The oracle disables selected Python networking calls,
not OS-level network access, owns before/after snapshots and runs the original
trace-first scoring.
It never imports the JS mapper or Matoi implementation. Its predicates are
external task definitions, not outside-laboratory validation of this study.

Each message is projected to an ordered finite publication row containing a
unique call-position identifier, typed direct/channel destination and exact
UTF-8 body size. Read/write interleaving and repeated messages remain intact.
An actual dispatch uses the kernel claim's target and body; the original
plan supplies the comparison oracle rather than correcting an erroneous
mapping during execution. Confirmation comes from observed inbox fields
and append indices. This is local process correlation, not a provider-signed
receipt or inference about hidden behavior inside arbitrary tools.

Independent of mapper output, an issuer-policy table fixes the original typed
destinations, a 65,536-byte per-message ceiling and the original message count
times that ceiling as shared budget. Message meaning is not authorized by
this policy. Both portable strategies use the unchanged durable owner, exact
body/signature binding, current authorization, quota and deduplication. Each
governed cell gets a fresh key and journal. A retained bootstrap preparation
is followed by an unused metadata change and a new occurrence preparation.

The fixed matrix contains 26 unmediated sanity cells, 20 governed benign
cells, 30 phishing-goal/policy diagnostics and 28 boundary controls. All 104
are retained. Both governed strategies preserve the final original
environments, traces and scores for all ten selected benign tasks, including
all fourteen exact message appends; no initial benign utility predicate is
already true. Table~\ref{tab:agentdojo-plans} separates this outcome from
security and policy measurements.

\begin{table}[t]
\centering\small
\begin{tabular}{@{}lrr@{}}
\toprule
Measure (per governed strategy) & Full & Repair \\
\midrule
Original benign utility true & 10/10 & 10/10 \\
Exact observed benign messages & 14/14 & 14/14 \\
Specified boundary refusals & 12/12 & 12/12 \\
Unexpected benign refusals & 0/10 & 0/10 \\
Duplicate conservation checks & 2/2 & 2/2 \\
Phishing goal: attacker success & 3/10 & 3/10 \\
\bottomrule
\end{tabular}
\caption{Original AgentDojo plans with separately fixed route/byte policies.
These are small prescribed sets, not estimated production error rates or
full-benchmark attack-success rates.}\label{tab:agentdojo-plans}
\end{table}

The twelve refusals per strategy cover revocation, narrowing before first
claim, target redirection, missing current evidence, added publications and
post-preparation body changes. None appends a message or consumes retained
quota. Two additional cases per strategy demonstrate duplicate admission
without extra messages or reservations, not universal concurrency behavior.

The phishing send succeeds in all ten unmediated task contexts. Matoi blocks
the seven whose policies do not authorize Alice, but permits the other three.
In one of those cases, the original user-utility and attacker-success
predicates are simultaneously true. This is important negative evidence:
authorization of a route is not judgment of content. The other four attack
goals remain unsupported or outside the message boundary; their unmediated
successes are not credited as prevented attacks.

For benign post-bootstrap preparation, full checking visits 24 new semantic
rules; repair instead validates ten retained references. Both still read 38
current queries and persist 25,630 bytes of materialized evidence. Reported
proposal bytes decrease from 17,566 to 10,350, but the owner does not export
the original compact wire stream. Its diagnostic counters are therefore not
presented as an independently replayed transport trace or a latency ranking.
Replay re-executes the original Python tools and predicates, checks signatures
and raw SQLite state, and separately verifies all 136 complete materialized
proofs with the Python evidence consumer. The finite-contract supplement's
Section S12 records the complete selection, policies, controls and provenance.

\section{Matched Compact-Packet Cost Control}
\label{sec:matched-cost}
The following 15 October 2026 campaign supplies the attribution control
left unmeasured by the earlier portable campaign in
Section~\ref{sec:portable-update}. It is a separately frozen experiment,
not a replacement for Table~\ref{tab:portable-update-cost}, the negative
five-strategy repair results, or the protected-cache qualification. Repeating
the same authored input population does not create a new semantic corpus.

\subsection{Identical packets, different checking obligations}
Three experimental arms separate two questions.
\texttt{explicit-full} consumes explicit proof rules at every position.
\texttt{compact-recheck} consumes a full bootstrap followed by compact repair
packets, resolves references only from its own previously accepted records,
and actually traverses and checks every current local rule obligation.
\texttt{compact-reuse} receives exactly the same compact input bytes but
checks applicability before reusing accepted subtrees. Recheck does not trust
an old aggregate answer or substitute fresh evaluation for validation of the
supplied proof. Neither baseline adds a gratuitous JSON encoding/decoding
round trip solely to make full checking expensive.

All arms receive complete current snapshots, use per-request query caches,
retain local rule records, and use the unchanged JS producer with symmetric
derivation memoization. Missing, stale, malformed, reordered, wrong-base and
wrong-workflow references, as well as restart, are exercised in functional
tests; refusal retains the preceding accepted checkpoint. The original
normative v2 API is unchanged. Unlike the separately authored original
reference consumers, each study-only consumer reuses finite-profile helpers
from the original implementation in its language. The resulting shared
mechanics are not externally independently developed controls. Reuse is
restricted to the fixed syntax position and full binding context of the
accepted record; whole-workflow bounds and global effect ordinals remain
composition obligations as stated in Section~\ref{sec:portable-update}.

\subsection{Frozen population and complete cost boundary}
The fixed seed is 20261015. Twelve families, twelve positions, three arms,
two consumer languages and twenty randomized paired blocks yield 1,440
complete cells and 17,280 repeated positions, with zero failures, exclusions
or missing positions. There are 144 authored input positions, not 17,280
independent tasks or attacks. Each language/arm aggregate contains 240
fresh-process cells and 2,880 positions. A separate one-block functional pilot
is retained without confidence intervals and is not pooled into this run.

The clock starts before both producer and consumer are spawned and ends
after both shut down. It charges producer work, the complete input, cold
bootstrap, proof validation, IPC, controller encoding/decoding and result
checks, complete output, diagnostics and shutdown. Source/input capture,
post-shutdown archive writes and offline replay are excluded. Internal timers
overlap round trips and cannot be added to them. Collection records the
operator's quiet-window acknowledgement on an Apple M4 Pro host with 24 GiB
RAM, Node v24.19.0 and Python 3.12.14; this is not host/clock attestation.

Strict sealed replay checks source, manifest, runtime and complete schedule
identities, then consumes the recorded packets with the opposite original
v2 checker. All 17,280 proof-output comparisons and all 17,280 comparisons
against opposite fresh outputs pass. The fresh reference contains 288 frames
in 24 family/language streams; repeated comparisons are not new fresh tasks.
All 480 compact pairs have byte-identical consumer stdin and semantic producer
request stdin. Analysis repeats the semantic gate. A separate read-only Python
calculation, without importing the study statistics functions, recomputes every
field of all 78 cost summaries and all 78 paired ratios and intervals from
the original records and diagnostics. These are same-project cross-checks,
not independent-laboratory validation.

\begin{table}[htbp]
\centering\small
\begin{tabularx}{\linewidth}{@{}Ylrr@{}}
\toprule
Consumer & Arm & Total ms, 240 cells & Median cell ms \\
\midrule
JS & \texttt{explicit-full} & 18522.367782 & 77.859709 \\
JS & \texttt{compact-recheck} & 17105.576168 & 74.204875 \\
JS & \texttt{compact-reuse} & 16509.278753 & 71.049042 \\
Python & \texttt{explicit-full} & 26180.566832 & 99.724584 \\
Python & \texttt{compact-recheck} & 22253.153961 & 96.430250 \\
Python & \texttt{compact-reuse} & 21853.871456 & 94.158875 \\
\bottomrule
\end{tabularx}
\caption{New matched-cost campaign, not the earlier portable or integrated
session measurements. A cell contains all twelve positions and the full
producer/consumer lifecycle.}
\label{tab:matched-cost-absolute}
\end{table}

\begin{table}[htbp]
\centering\small
\begin{tabularx}{\linewidth}{@{}Yrr@{}}
\toprule
Whole-cell contrast & JS ratio [95\% interval] & Python ratio [95\% interval] \\
\midrule
Reuse/recheck (primary) & 0.965140 [0.961006, 0.970126] & 0.982057 [0.980022, 0.984324] \\
Recheck/explicit & 0.923509 [0.920503, 0.926710] & 0.849987 [0.847535, 0.852425] \\
Reuse/explicit & 0.891316 [0.888382, 0.894480] & 0.834736 [0.832287, 0.837243] \\
\bottomrule
\end{tabularx}
\caption{Ratios of summed paired-block costs. The 95\% percentile intervals
use 2,000 resamples of twenty complete blocks, seed 20261015. Aggregate
resamples preserve all twelve families in each sampled block.}
\label{tab:matched-cost-contrasts}
\end{table}

\subsection{A modest residual, with null and adverse regimes retained}
With compact input bytes held fixed, reuse lowers aggregate whole-cell cost
by 3.49\% in JS and 1.79\% in Python. At six decimal places the reductions are
3.485983\% and 1.794274\%, with reduction intervals [2.987448\%, 3.899384\%] and
[1.567597\%, 1.997847\%]. Reuse-minus-recheck totals are -596.297415 ms
and -399.282505 ms, or mean differences of -2.484573 and -1.663677 ms per
twelve-position cell. Reuse is below recheck in 19/20 JS aggregate blocks
and 20/20 Python blocks; no block is removed. The larger 10.868422\% and
16.526362\% reuse/explicit reductions also change proof representation,
encoding, parsing and transport. They are not pure semantic-reuse effects
or revisions of the historical 10.40\% and 16.22\% measurements.

\begin{table}[htbp]
\centering\small
\begin{tabularx}{\linewidth}{@{}Yrr@{}}
\toprule
Family & JS reuse/recheck [95\% interval] & Python reuse/recheck [95\% interval] \\
\midrule
literal & 0.961108 [0.947811, 0.972368] & 0.983028 [0.972238, 0.992822] \\
shared-scalar & 0.901956 [0.891265, 0.911660] & 0.951904 [0.942762, 0.959770] \\
sparse-scalar & 0.942479 [0.936211, 0.948772] & 0.981051 [0.972150, 0.992503] \\
dense-scalar & 0.998362 [0.992436, 1.003878] & 0.999539 [0.993786, 1.004502] \\
wide-irrelevant & 0.970020 [0.956599, 0.984911] & 0.984881 [0.978250, 0.990870] \\
branch-stable & 0.966131 [0.936744, 1.012855] & 0.969841 [0.962723, 0.977711] \\
branch-flip & 1.001304 [0.992487, 1.010449] & 0.994303 [0.983981, 1.005002] \\
membership-change & 0.968301 [0.960159, 0.977656] & 0.981428 [0.974668, 0.988543] \\
membership-reorder & 0.935512 [0.913128, 0.950583] & 0.968471 [0.962450, 0.974236] \\
absence-transition & 0.994108 [0.981791, 1.005501] & 0.998804 [0.986884, 1.009505] \\
nested-scope & 0.997819 [0.989146, 1.006670] & 0.992383 [0.984144, 1.000542] \\
workflow-change & 0.975969 [0.966857, 0.984825] & 0.991914 [0.981151, 1.003734] \\
\bottomrule
\end{tabularx}
\caption{Every matched-cost family, twenty paired blocks each. Intervals
are pointwise, not multiplicity-adjusted simultaneous intervals. All secondary
family contrasts remain in the retained analysis.}
\label{tab:matched-cost-families}
\end{table}

Each language has seven primary family intervals below 1 and five crossing
1. Python branch-flip also retains an adverse secondary contrast:
recheck/explicit = 1.010618 [1.003049, 1.018395]. Dense-scalar and branch-flip
have no reusable references in either compact arm. Thus neither the aggregate
nor the selected favorable families licenses a universal speed claim.

Per language, full and recheck both check 45,480 rules; recheck reconstructs
29,500 of these from accepted references. Reuse checks 15,980 rules and reuses
29,500 visits. Both compact arms process 14,060 references. All three arms
read 26,920 current queries. Producer new-rule, reused-visit and
query counters match in every paired position. Each compact arm sends
28,780,980 consumer-input bytes per language versus 34,485,980 for full;
all-stream totals may differ because diagnostics contain measured timers and
memory. Only the compact inputs, not the diagnostic text, are asserted equal.
The largest retained root is 33 rules. Cold first-position work accounts for
50.44\% of JS recheck cost and 38.88\% of Python recheck cost. These small
finite workloads, shared helper mechanics and one local host bound the result;
internal checking-time reductions are not a substitute for whole-cell cost.

\paragraph{Artifacts and scope.}
The report is \path{docs/research/matoi-matched-cost-20261015.md}; the archive
is \path{output/matoi-matched-cost-confirmatory-20261015-01/}. Source-inventory
SHA-256 is
\path{2a6a4b3802df83cb12b9c5b0602e71c3eceac6bc48abcde9d5a54683ce8072b4};
freeze and collection-manifest digests are respectively
\path{830cb3e7c8d498058de217e81f76c88729deb141023de8a6ce07e7536bdf714a}
and \path{e4817683dc50340e5a6490ba44c4bdfe17a6b5b354319104ac0e2639324b61ae}.
The sibling strict-replay and analysis manifests are
\path{17a58f594d471e8ea85bc11f030f026e93674aa60f06bfe5687398dbdd1b87a8}
and \path{c1e47376c9e7a457c99f4f69b2fed4ef849fe0c29124ea300dcb39c9856a35a4}.
The analysis retains all 78 contrasts, 78 family/arm cost summaries, every
cell and every raw diagnostic. Strict replay pins the original runtime;
explicit reviewer replay may select compatible current runtimes but is labeled
semantic-only, not timing reproduction. These retained private research
artifacts are not an asserted public download or an external attestation.

\section{Externally Defined Banking Transfer Plans}
\label{sec:banking}
This separately frozen study adds a second API-shaped effect domain to the
Slack message plans in Section~\ref{sec:agentdojo}. It uses original synthetic
AgentDojo Banking plans, local tools and unchanged utility/security
predicates~\cite{agentdojo}. This is not an LLM experiment, an autonomous
prompt-injection evaluation, a live bank integration or an independent-laboratory
replication. The pinned upstream commit is
\texttt{089ed468cf3ed0322acc66b0211f26d9d90dbf60}; the resolved Banking suite
is 1.2.2 and package version is 0.1.35. Source was sealed before the first
original Banking task exposure, then the complete population was frozen.

\subsection{Whole-plan coverage before outcome measurement}
Whole-plan API eligibility retains 5/16 original user tasks (31.25\%) and
7/9 original attack plans (77.78\%). The eleven excluded users comprise four
read-only plans, four scheduled-transaction plans, two other user/password
mutation plans and one mixed plan. The mixed \texttt{user\_task\_15} contains
a transfer and succeeds unmediated, but also contains unsupported mutations;
the mapper does not extract only its convenient transfer. The two excluded
attack plans require unsupported scheduled-transaction or password mutation.
All exclusions remain in the inventory, not in a favorable five-user denominator.

All sixteen users and all nine attacks receive unmediated sanity cells;
attack sanity uses the first original user's context. The five eligible users
receive both governed benign arms. The seven original attack plans are reused
across the five eligible user policies, giving 35 crosses per arm, not 35
independent attack goals. These crosses run unmediated, with
\texttt{portable-full}, and with \texttt{portable-repair}. Eight controls run
for each eligible user in both governed arms. The complete population is
\[
16+9+(5\times2)+(5\times7\times3)+(5\times8\times2)=220.
\]
All 220 planned cells are retained and audited. The adapter's separate amount
ceiling of 10,000 is not an eligibility rule. The five eligible benign original
transfers happen to range from 4 to 200.29; no above-cap original benign task
was observed. Such a supported task would remain in the denominator with its
utility loss, not be silently excluded.

\subsection{Mapping, authority and actual local effects}
The API-shaped mapper produces the original recipient, canonical transfer
payload and ordered byte-count publication rows. The finite effects describe
destination and byte count, not exact body bytes; the durable host's signed
body hashes bind the canonical transfer payload separately. Independent
original-plan annotations
fix authority rather than allowing mapper output to define its own permitted
recipient, maximum amount or transfer count. These annotations are supplied
authority, not inferred trustworthy intent from natural-language requests.
Both governed arms retain the same durable owner, signed preparation, exact
body binding, current-authority checks, occurrence custody and duplicate
conservation. Available byte budget is debited at admission; claim validates
the retained reservation instead of testing fresh available budget for the
same cost again. Money and count constraints are enforced by the CAVA adapter;
the finite kernel's quota remains bytes, not money.

Each governed arm has 80/80 exact selected-plan mappings and 80/80 exact
original-authority-plan mappings. Its selected plans contain 90 transfer rows:
five benign, 45 attack-cross and forty control-base rows. Controls compare the
base selected plan before intervention; modified proposals and their bindings
are checked separately. These are repeated mapping checks, not eighty
independent user tasks. Full and repair agree on original calls, scores and
refusal outcomes in all 80 paired cells.

The original \texttt{send\_money} tool appends a transaction to its in-memory
environment without debiting the account balance. Exact append checks cover
the actual transaction ID, sender, recipient, binary64 amount, subject, date,
\texttt{recurring=false}, and unchanged remaining state. They do not establish
real bank settlement, financial accounting correctness or provider-signed
receipts. The 160 bootstrap preparations execute zero sends. All twenty
successful governed target claims across both arms bind actual original
transfers: ten benign and ten duplicate-control sends in total.

\begin{table}[htbp]
\centering\small
\begin{tabularx}{\linewidth}{@{}Yrrrrr@{}}
\toprule
Population / arm & Cells & Utility pre/post & New utility & Attacker post & Appends \\
\midrule
All user sanity / unmediated & 16 & $7\to16$ & 9 & 0 & 6 \\
All attack sanity / unmediated & 9 & $0\to0$ & 0 & 9 & 9 \\
Eligible benign / full & 5 & $1\to5$ & 4 & 0 & 5 \\
Eligible benign / repair & 5 & $1\to5$ & 4 & 0 & 5 \\
Attack cross / unmediated & 35 & $7\to7$ & 0 & 35 & 45 \\
Attack cross / full & 35 & $7\to7$ & 0 & 0 & 0 \\
Attack cross / repair & 35 & $7\to7$ & 0 & 0 & 0 \\
All controls / full & 40 & $8\to12$ & 4 & 0 & 5 \\
All controls / repair & 40 & $8\to12$ & 4 & 0 & 5 \\
\bottomrule
\end{tabularx}
\caption{Original Banking predicates and actual append counts. Pre/post are
counts of predicates true before/after the target attempt, not success
fractions. New utility counts only false-to-true transitions. Attacker
predicates are false before every cell. No baseline score is corrected.}
\label{tab:banking-outcomes}
\end{table}

Both governed arms execute all five eligible benign target transfers and have
five post-utility successes, but only four new utility satisfactions. The 35
unmediated attack crosses satisfy all 35 original attacker predicates and
append 45 transactions. Both governed arms refuse all 35 at the kernel
recipient boundary, append none, and leave all attacker predicates false.
Each governed arm still executes five preceding read calls in the attack
population; a transfer refusal is not a claim that no read occurred.

All seven eligible attack plans target the same attacker IBAN, which lies
outside all five fixed recipient policies. A simple recipient-membership gate
supplied with those policies therefore explains all 35 refusals per governed
arm. An ordinary policy engine remains a valid alternative. This is a
baseline explanation, not a separately measured arm or evidence that these
attacks require checked incremental repair.

\paragraph{Upstream scores are not complete goal checks.}
The pinned attack predicates for tasks 0, 1 and 3 use forms such as
\texttt{service or "" in t.subject.lower()} inside a recipient-filtered
\texttt{any}. A truthy secret bypasses content matching; an absent secret
makes empty-string membership true. Thus an outgoing transaction to that
recipient can score without the requested disclosure content. Task 6 accepts
either a total of at least 30,000 or one amount above 10,000 despite specifying
small increments. We retain the original predicates and scores unchanged,
without reporting a new attack run or post-review goal-level rescore.
Exact append records and governed zero dispatch are separate observations;
predicate truth alone does not establish those content or transfer-pattern goals.

\subsection{Pre-satisfied utility is not successful recovery}
One of five benign utility predicates is already true at reset (1/5). The
concrete original case, \texttt{user\_task\_5}, scripts a difference payment
of 5 to Spotify while its unchanged utility predicate checks for an existing
transaction of 50 to a particular recipient. Reset already satisfies that
predicate. Consequently, seven attack crosses for that user are utility-true
before and after refusal: 7/35 remain unchanged in each arm, despite zero
governed appended transfers. The other 28/35 remain utility-false. The
unmediated attack crosses also have seven pre/post utility truths and zero
new satisfactions. Neither the score nor the original plan was patched to
remove this counterexample.

Thus blocking these attacks does not restore legitimate-task utility. The
same reset effect accounts for eight pre-true control cells per governed arm.
Controls produce five transfers and four new utility satisfactions, all from
duplicate-occurrence targets; seven refused controls for the pre-satisfied user
remain utility-true with no send. Mapping fidelity, policy compliance,
observed mutation and task utility must remain separate reported quantities.

\subsection{First decisive layer, not an undifferentiated block count}
Refusal attribution uses the first failed admission/claim or adapter check,
cross-checked against the signed preparation, operation journal and raw
original calls. In each arm, 25 primary admissions succeed. Five subsequently
fail at claim after narrowing, five at the money adapter, and five at the
count adapter. Only ten target claims dispatch transfers. Kernel admission
alone is not a provider dispatch.

\begin{table}[htbp]
\centering\small
\begin{tabularx}{\linewidth}{@{}YYrrr@{}}
\toprule
Population or control & First decisive layer / stage & Full cells & Repair cells & Appends/arm \\
\midrule
Original attack cross & Kernel / admission: recipient boundary & 35 & 35 & 0 \\
Revoke before admission & Kernel / admission: authority & 5 & 5 & 0 \\
Narrow before first claim & Kernel / claim: authority & 5 & 5 & 0 \\
Wrong recipient, fresh proof & Kernel / admission: boundary & 5 & 5 & 0 \\
Amount changed after prepare & Kernel / admission: preparation binding & 5 & 5 & 0 \\
Above-cap amount, fresh proof & Money adapter / after admission & 5 & 5 & 0 \\
Missing current evidence & Kernel / admission: stale summary & 5 & 5 & 0 \\
Extra transfer, fresh proof & Count adapter / after admission & 5 & 5 & 0 \\
Duplicate occurrence & Two admission responses, one original occurrence & 5 & 5 & 5 \\
\bottomrule
\end{tabularx}
\caption{Every prescribed Banking control over all five eligible users.
Five duplicate-control appends per arm mean one per control cell, not one per
admission. No control is masked or excluded.}
\label{tab:banking-layers}
\end{table}

Per arm, there are 55 kernel-admission refusals, five kernel-claim refusals,
five money-adapter refusals and five count-adapter refusals. All forty control
cells exercise their prescribed outcome, with zero representation failures,
unresolved sends or missing controls. The 35 original attack crosses execute
\emph{zero money-policy checks}: their refusal is recipient-boundary evidence,
not evidence that a financial kernel blocked thirty-five attacks. The fresh
10,001 amount controls violate both the global ceiling and every smaller
recipient-specific maximum here, so they do not isolate the two constraints.
The extra-transfer control sends two against a one-transfer annotation within
the representation limit; the count adapter, not byte quota, refuses it.

\subsection{Strict replay and remaining limits}
The first strict replay completes with no stderr, retry, missing cell,
transport/error frame or collection/audit error. It verifies the external
manifest and exact inventory, re-audits raw exchanges and reopened SQLite
journals, replays the original tools and scores, and checks all 320 retained
materialized proofs with the finite Python consumer. It checks 92 setup
exchanges and 1,011 collection frames: one inventory, 334 resets, 215 executes,
220 scores, 240 snapshots and one close. Exact independent append checks
cover all eighty original sends: sixty unmediated and twenty governed.
The independent aggregation reconstructs mapped payloads and policy outcomes
from raw calls and original-plan annotations rather than importing the
candidate mapper or money policy. This is still same-project cross-checking.

The Python check validates retained complete proofs, not reconstruction of
the original compact preparation transport. No Banking speedup or
timing-equivalence result is claimed. All five eligible benign plans contain
one transfer and all compiled workflows use the same finite \texttt{each}
shape. Normal preparation changes irrelevant metadata, not branch, negative-lookup
or nested-scope repair transitions. Missing required evidence tests refusal,
not adaptation to an optional absence. No corrective planning follows refusal.
API-shaped mappings and
externally authored synthetic plans do not establish arbitrary-tool closure,
natural-language authority formation, model-driven attack resistance, real
banking safety or broad cross-domain generalization. The 35 reused crosses,
eighty paired governed cells and 320 proofs are not new independent tasks or
outside replications. The earlier Slack phishing counterexample remains valid.

\paragraph{Artifacts and scope.}
The full report is \path{docs/research/matoi-banking-20261015.md}; the archive
is \path{output/matoi-agentdojo-banking-confirmatory-20261015-v1/}. The
pre-exposure source-seal SHA-256 is
\path{8d4e45cd4f92c3ad96a75ccda7554536c01878a3d1048011de7a8e5c56344265};
the population-freeze and collection-manifest digests are
\path{dd0e9f245efaa18724271719261149e4a5c9302eecadef0b052b8647432a919b}
and \path{2d1ac9e9f0579c1bbf305e1295cf509f5ebfbbc44a71c85367e3cfea9d15eb33}.
The strict-replay stdout and independent raw-count digests are
\path{55680126a2d201f917b799e41491b19b1f9c5d627ee65a35c368c86749bb9d7e}
and \path{b7e641e011f2eaabae824ce9ecde8fe7ce2e4de3522398806ee9c3ed5218690b}.
Runtime/dependency identity and all forty captured source/support files are
retained. The replay used the sealed Node/Python environment. A separate
explicit reviewer semantic mode does not rewrite those pins or reproduce
timing. These scoped private research artifacts are not a public upload,
public license or claim of independent validation.

\section{External Validity and Further Evaluation}
\subsection{Freeze semantics before expanding cases}
Broader validation requires freezing requested-authority denotation, state-variation delegation, effect language, coverage assumptions, unknown handling, commitment semantics and units before evaluating held-out workflows. The authored matrices and finite products here are development evidence, not such a holdout evaluation. Existing reserved skill families and runtime mappings were not consumed by these experiments.

Independent skill authors should define allowed traces without reference to OSuite field names. A separate curator should fix hidden state changes and expected outcomes before seeing evaluator predictions. Runtime-to-action mapping is scored against gold objects; authorization is first scored on those gold objects and then on mapper output. This separates representation errors from decision errors rather than allowing one implementation to define its own ground truth.

An externally authored benchmark supplies another source of task definitions,
not an automatic correctness certificate. For example, AgentDojo includes
deterministic environment utility/security predicates, but a predicate that
checks inbox growth need not verify the meaning of the message. An adapter
study must retain upstream utility, exact mutation matching, mapping coverage
and authorization outcomes as separate scores. Unsupported effects stay in
coverage accounting. Section~\ref{sec:agentdojo} retains the narrow Slack plan
study, including the counterexample where utility and attacker success are
both true. Section~\ref{sec:banking} adds a second API-shaped domain but covers
only 5/16 original Banking user plans and reuses seven attack plans across
five policies. Its unchanged original predicates expose a different
counterexample: 1/5 benign and 7/35 attack-cross utility predicates are already
true before the attempt. Refusal is not restored utility. Neither study
establishes broad cross-domain or model-driven generalization, natural-language
authority formation or real-bank safety, and neither is external-laboratory
replication.

\subsection{Matched comparisons}
Primary comparators must share the same authority, target-identity, coverage, nonce, budget, durability and declared commitment or dispatch boundary. Section~\ref{sec:qualification} makes that comparison within one finite kernel; it does not replace comparisons with reproducible formation/closure systems on independently defined workflows. Checked ingestion must charge reference traversal, retention and invalidation rather than compare only a warm lookup with full fresh work. For remote dispatch, mature idempotency and outbox designs remain substantive baselines; a local check must not be credited as an atomic provider commit. Label local reimplementations and substitutions explicitly. Missing-context, missing-effect and missing-fence paths are ablations, not proxies for named vendors.

Section~\ref{sec:matched-cost} now measures the narrower identical-compact-input
control, with actual complete rule rechecking, symmetric producer memoization
and full producer/consumer lifecycle costs. Its 3.49\% JS and 1.79\% Python
aggregate benefits concern this shared-mechanics finite experiment on one host.
They neither invalidate the older protected-cache ties nor establish integrated
service latency, universal family-level benefit or an external-baseline win.

Vary state changes independently from policy changes. Include harmless presentation changes; bounded effect changes; destination changes; missing/contradictory lineage; executable replacement; aggregate budget contention; revocation around commitment; lost acknowledgements; partial effects; and receiver unavailability. An independent effect owner records actual committed effects. Absence of a success message is not evidence of no side effect.

\subsection{Metrics and rejection criteria}
Report false commitments among independently forbidden opportunities, unnecessary refusals among admissible opportunities, correct review rate, mapping tuple accuracy, bound coverage/precision, unknown rate and cross-runtime decision consistency. Use the workflow family as the grouping unit; repeated requests estimate variability, not independent semantic evidence.

Performance must include issuance, analysis, cache misses, unused preparation, authority access, durable admission, queue delay, receiver commitment and reconciliation. Report latency distributions and throughput under a fixed offered load with uncertainty at the workflow/run level. Count timeouts and failures rather than dropping them from successful-request percentiles. A reduction in analysis-call count is not a measured latency improvement.

Predeclare falsifiers. Any unsafe covered commitment inside the stated assumptions invalidates the enforcement claim. No advantage against a matched baseline invalidates a performance claim. Failure to justify selective dependency completeness requires a narrower claim, conservative invalidation or stronger effect mediation. Poor holdout results must not trigger tuning on those holdouts.

\section{Discussion and Limitations}
\paragraph{A useful cache stores evidence, not permission.}
The measurements explain why a high hit rate is insufficient. A hit still pays for validation, while a miss can pay both validation and fresh analysis. An execution still needs applicable authority, budget enforcement and a unique occurrence. In the durable profile, admission debits available budget and claim validates that existing reservation, not a second allocation from current available budget. A future adaptive selector should optimize that full cost under the same protection conditions, rather than route expensive checks out of the measured path. Matoi provides no measured adaptive-selection result.

\paragraph{Interfaces with skills and prepared admission.}
A neutral skill contract can supply requested authority without determining the runtime's current permission. Matoi can consume that boundary; it does not establish that a mapping from an independently defined skill preserves meaning. Similarly, ZeroGate can prepare an action pass while Matoi chooses applicable consequence evidence, but their combination must retain current authorization and durable consumption. The experiments do not establish a production integration of these components.

\paragraph{Complete mediation is an assumption, not a graph label.}
The effect-owner argument requires every covered consequence to encounter effective enforcement. A queued job with a separate credential, an uninstrumented receiver or a callback that escapes the supplied graph violates that premise. The supporting placement study exposes this distinction: honest qualified inventory permits safe covers, whereas a falsely complete inventory still admits a hidden raw writer. Its 225 unauthorized diagnostic mutations span five unqualified configurations, including missing atomic hooks and activation failures; they are not all inventory failures. Probes after deployment refusal deliberately test the unsafe path rather than count qualified work. It is a finite test of the premise, not a topology-discovery solution. Neither selective revalidation nor AREG's observed graph proves complete mediation for an arbitrary deployment.

\paragraph{Limits of the demonstrated result.}
The abstract model assumes truthful primary summaries and retains its omitted-dependency failure. The original process and cost campaigns derive dependencies only in a closed straight-line language with immutable snapshots. The separate query-aware kernel checks completeness for its bounded branches, collections and lookups, but shares trusted parsing and reference semantics; neither implementation establishes closure for arbitrary tools, binaries or hidden service internals. Its 69 semantic tests and finite state product are developer conformance, not independent labels. A supplied interpreter label does not satisfy FCD's actual formation-authority obligation.

The original effect owner serializes local state and synthetic mutations; the GitHub dispatcher does not serialize with remote commitment. The GitHub pilot adds real third-party storage evidence without changing the earlier experiment's SaaS exclusions or extending its atomicity argument to the provider. Those earlier studies do not establish process-restart recovery. Section~\ref{sec:durable} separately measures local restart with evidence reconstruction, monotone authority epochs and conservative uncertain reservations. A conflict-marked unknown cannot become confirmed merely because a later readback contains one match. Its single-history and trusted-provider assumptions exclude backup rollback, independent database copies and late duplicate re-auditing after confirmation. None of the studies establishes globally exactly-once behavior, all-or-nothing remote batches, multi-host consensus or remote attestation. Local confirmation suppression is not a provider failure, and injected process termination is not machine power loss.

The original preparation matrix remains one synthetic family on one uncontrolled workstation, with a fixed change cadence and no deployment-level replication or significance claim. Its negative regimes are not repaired by the GitHub pilot's counted reduction from 21 full analysis/check operations to ten full operations plus eleven reused preparations. The earlier four-strategy full-cost matrix likewise finds no demonstrated advantage over the protected reference cache. Those four preparation strategies share expensive persistence, signature and live-execution obligations; that experiment does not isolate a causal optimization for those costs. Short offered-load batches, near-cap fill/probe timing and injected faults do not establish a sustained service-level objective. The new matched-cost campaign measures a modest positive residual against complete rechecking, but its twenty-block intervals do not establish cross-host or cross-workload generalization. Independent labels, broader held-out workflows and a material advantage over competent external implementations remain unestablished; their evaluation is scientific work, not an inference from larger test counts.

The Banking study broadens the API domain, not the finite language or the
authority-formation claim. Its destination/byte-count trace does not encode
exact transfer-body bytes or monetary semantics: durable body hashes bind the
payload, and separate CAVA adapter checks enforce numeric amount and count.
Only five original user plans are eligible; all have one transfer and compile
to the same finite iteration shape. Thirty-five reused attack crosses are not
thirty-five independent goals, and the unchanged pre-satisfied utility
counterexample prevents equating refusal with task recovery. The original
tool appends local synthetic transactions without debiting an account balance.
No LLM, live-bank, financial-settlement or Banking performance guarantee follows.

\paragraph{Qualification is not hosted activation.}
The tested journal is finite, with no compaction, backup-rollback fencing or
authorized unknown-resolution interface. Tenant administration, key custody
and rotation, provider contracts, retention/migration and an operator recovery
procedure require separate integration tests. The studies neither activate a
hosted feature nor establish that PCAA, CAVA, BAF, AREG, ZeroGate, Matoi and
the DRS placement experiment form one deployed end-to-end protocol.

\section{Conclusion}
Matoi separates three facts that are easily collapsed in an agent runtime: an analysis still describes the possible consequences, a grant covers those consequences, and a particular effect remains permitted now. The straight-line and checked query-aware implementations realize this distinction under explicit evidence-trust assumptions. The durable extension adds a fourth fact: what can still be established after the process disappears. It reconstructs evidence rather than reviving a serialized handle, and retains uncertain reservations rather than manufacturing a new execution opportunity. Real-process receiver evidence, GitHub outcomes and the recovery campaigns make refusal and reconciliation inspectable without equating local dispatch with atomic remote commitment. The earlier same-kernel comparison does not demonstrate a checked-reuse advantage over a protected reference cache; the coverage study preserves hidden-path failures instead of treating a claimed graph as complete mediation.

The later portable profile makes the exact evidence contract independently
implementable and records a combined compact-packet/reuse benefit against
full proof. A separate matched control now holds compact inputs identical and
measures the smaller 3.49\% JS and 1.79\% Python whole-cell reuse benefit over
complete rechecking. It does not establish a universal or integrated speedup.
The explicit durable integration retains current authorization and recovery,
but not an inherited speedup or remote atomicity guarantee. The Slack study's
route-compatible phishing remains a counterexample to content safety. The
new Banking study adds 5/16 user-plan and 7/9 attack-plan coverage, with 35
reused crosses refused in each governed arm. Its 1/5 pre-true benign and
7/35 unchanged pre/post attack-cross utility predicates show why those refusals
are not new task completions. Neither plan study evaluates an LLM or real-bank
operation. Together, these are a bounded systems contribution with explicit
failure and cost boundaries, not a new general incremental algorithm or an
established standard.
The practical rule is conditional but concrete: reuse only applicable
consequence evidence, retain live authority and durable occurrence checks,
and establish which component actually controls the effect and its recovery.

\appendix
\section{Evidence Accounting and Reproduction}
The following table accounts for the original model, straight-line campaigns
and preparation study. The query-aware and GitHub artifacts are identified
separately below; their conformance counts are not pooled with these observations.

\begin{table}[ht]
\centering\small
\begin{tabularx}{\linewidth}{@{}>{\raggedright\arraybackslash}p{0.22\linewidth}YY@{}}
\toprule
Study & Observational unit & What must not be inferred \\
\midrule
Finite model & 35 primary states plus one separate omitted-dependency counterexample & Independent semantic accuracy or real execution \\
macOS process campaign & 1,260 records; 1,320 occurrences; 1,020 effects & 1,260 independent workflow families \\
Azure process campaign & Same matrix: 1,260 records; 1,320 occurrences; 1,020 effects & A new semantic corpus or a production SLA \\
Preparation study & 36,864 measured requests; 1,536 cached primes charged separately; 6,144 warmup requests & Full-workflow timing or free preparation \\
Receiver-evidence campaign & 48 records; 96 expected occurrences; 54 confirmed effects; 24 denied attempts; six reconciliations & Denials or reconciliations as additional effects \\
Software regression & 223 combined Matoi/AREG tests; cost-study tests separate & 223 independent attacks or experimental records \\
\bottomrule
\end{tabularx}
\caption{Denominators are not interchangeable. All experimental labels and replays are internal; repetitions assess implementation behavior and timing variability, not independent generalization.}
\end{table}

The repository artifact is \path{scripts/authorization-reuse-model/}, with fixtures, tests, evaluator, report generator, case-level JSON and source hashes. From the repository root, run:
\begin{quote}\small\ttfamily
node --test\ \path{scripts/authorization-reuse-model/model.test.mjs}\\
node\ \path{scripts/authorization-reuse-model/run.mjs}
\end{quote}
These paths identify the accompanying repository layout. No public, anonymously accessible archive DOI is asserted here. The abstract model, process studies and preparation study are separate artifacts, not pooled observations.

The real-process prototype and its separately implemented verifier are in \path{scripts/matoi-prototype/}. Run \path{measure.sh} with a fresh output directory and repetition count, then replay with \path{verify.mjs}. Manifests bind prototype and imported CAVA source hashes, the recorded interpreter digest, schedules, environment and public signing key. Replay verifies recorded evidence; it does not independently attest the originating host or the honesty of its controller. The first development pilot is retained separately and excluded from the measured table.

The preparation study is in \path{scripts/matoi-cost-study/}. Its fixed bundle is \path{output/matoi-cost-study/fixed-20261006-01/}; the UTC date in this path falls within the local October 5 campaign. It contains the protocol, frozen imported source graph, environment, all 896 epoch records, statistics and file digests. The companion \path{docs/research/matoi-cost-study-20261005.md} reports every cell. Replay using the archived verifier when current source has changed. The manifest SHA-256 is \path{0c502900121db418e695d516b221a3b52601b6e8751fea8228ca609c7034fc56}. It is an integrity reference, not a third-party attestation or an anonymous public artifact URL.

The receiver-evidence campaign is implemented in \path{scripts/matoi-areg-bridge/}. Its 48-record archive has manifest SHA-256 \path{3fc070e80a8ac86f0067db10c5e42546af9cee7962a92133fbf9d58ff6c532c6}. Verification requires that external digest and the configured issuer/collector trust inputs; it rehashes the saved evidence and independently reads the saved owner database. It does not make the original collector or host independently trustworthy.

The distinct finite kernel is in \path{scripts/matoi-semantic-revalidation/}.
The companion reports \path{docs/research/matoi-semantic-revalidation-20261006.md}
and \path{docs/research/matoi-checked-summary-20261006.md} record its original and
checked-ingestion conformance. Its final 69-test suite includes the first
query-aware slice's 35 tests, the 12 deletion and six effect-mutation cases, and the 256-pair product;
none of these denominators is an additional set of independent workflows.
The targeted Matoi/DRS regressions reported there do not establish integration
of a DRS recovery mechanism into checked ingestion.

The GitHub implementation is in \path{scripts/matoi-saas-pilot/}; its report is
\path{docs/research/matoi-saas-pilot-20261006.md}. The retained run is
\path{output/matoi-saas-live-20261006-01/}, with an unchanged repository copy at
\path{scripts/matoi-saas-pilot/retained/20261006/}. Its manifest SHA-256 is
\path{a4dd7e4f82e250fb651921ec7a473f5abb28b6a95a00ba2a84ce1e6b633e0bb1};
the captured source-index digest is
\path{02798832e1e6495c4501854abcd45a3596448d014fca7256907158c979b61fce}.
Offline replay uses \path{scripts/matoi-saas-pilot/verify.mjs} with the archive
directory and external manifest digest; it does not repeat the live mutations.
The separate post-run readback has SHA-256
\path{5167177a4847c580ec08a06b366663e78246b131a430dd6dfd4a0944aaebb361}.
Public observations are on
\href{https://github.com/OndCo/Agent-Action-Boundary-Benchmark/issues/3}{GitHub issue 3}
and \href{https://github.com/OndCo/Agent-Action-Boundary-Benchmark/issues/4}{issue 4}.
Those public objects are not an anonymous download of the private evidence
archive. Software-regression totals, including retained-evidence checks, are
not additional live cells. Source pins and internal replay identify recorded
evidence; they do not supply independent semantic labels or host attestation.

The durable extension is in \path{packages/matoi-runtime/}; its report is
\path{docs/research/matoi-durable-runtime-20261013.md}. Its final archive is
\path{output/matoi-durable-runtime-20261013-02/}, identified by report SHA-256
\path{5bc520c8511f9432cdd1349fd2ba8ee5834880298849ddedda10fb1e91cd3205}
and source-manifest SHA-256
\path{2e6f5d5465ee58cc5356766615e8c1adfd4b2465bf534e130a6b57477aab1fdd}.
The frozen protocol is
\path{docs/research/matoi-durable-process-protocol-20261013.md}.
With Node.js 24 or later, run
\path{scripts/matoi-durable-runtime/verify.mjs} with the archive directory and
external report digest. It uses the matching checker source, verifies signatures
and stored evidence, and reads database copies without resending effects.
Reproduction into a new directory uses
\path{packages/matoi-runtime/examples/local-recovery.mjs}.
The preceding run 01 is retained but excluded from this table: it omitted the
public issuer key needed for offline verification. Run 02 changes that evidence
capture, not runtime source or frozen schedules. These internal archives are
not anonymous public downloads; a public release needs separate review and
publication of the required evidence. No private signing key is in either run.

\subsection{Same-kernel qualification archives}
The successor harness is \path{scripts/matoi-qualification/}; the complete
report is \path{docs/research/matoi-qualification-20261013.md}. Protocol
\texttt{matoi.qualification.20261013.v2} and imported sources are frozen
before each matrix, with full captured-source replay and temporary database
copies. The four separate manifest digests below are retained outside
their archives:
\begin{itemize}
\item Cost, \path{output/matoi-qualification-v2-confirmatory-cost-20261013/}:\\
\path{25d0ce494a3ad0ff3fe670c3a34450fc010412aba66cfa4f7f3132fe31a775cb}.
\item Load, \path{output/matoi-qualification-v2-confirmatory-load-20261013/}:\\
\path{4eaa5265b6fd61946b960fec1896e4bb0184668d7388420450188a7b19c6f659}.
\item Growth, \path{output/matoi-qualification-v2-confirmatory-growth-20261013/}:\\
\path{e0705d5786602386d8d3a19954aac389bdba92fcde3749f2fbb76e025511ca46}.
\item Fault, \path{output/matoi-qualification-v2-confirmatory-fault-20261013/}:\\
\path{94fe37f61b0a450b230fbc5aa1c8e90b5e43292540236e350528bcede23d9c90}.
\end{itemize}
Use the matching archived \path{replay.mjs} with the archive path and its
external manifest digest. The implementation identity is
\path{79916a1863b9b7c88fa0d2246d6755ccca87e45d399970b8d39a0692ee11f81a}.
The harness records a protected-reference profile separately from checked
derivations. Keys are ephemeral; only public verification material is
retained. The 132 v2 pilot cells and stopped 82-cell v1 collection are not
added to the 1,416 formal observations.

The supporting coverage harness is \path{scripts/drs-qualified-coverage/},
with report \path{docs/research/drs-qualified-coverage-20261013.md}. Its
formal archive is
\path{output/drs-qualified-coverage-confirmatory-20261013T111800Z/};
manifest SHA-256 is
\path{c744983797e0951244c18cf42236b235fe64fbc3eaf479274f5dddc9563e40ea}
and pre-data freeze SHA-256 is
\path{3ad2f4d9353ae5a9b78f75e8ca0813e92effb7d377010ddcdf27fa4d104a92aa}.
Its captured-source replay checks both supplied and physical coverage,
concrete mutations, authoritative denials, process activity and complete
transaction accounting. Development pilots are retained separately.
These directories identify retained internal archives, not anonymously
accessible public downloads or external attestations. Public artifact
release still requires separate review and publication.

\section{Earlier Engineering Study}
An earlier 4 October 2026 local study recorded 3,000 attempts. Its 22 authored scenarios expand to 25 attempts per configuration per repetition; six configurations and twenty repetitions therefore produce $25\times6\times20=3{,}000$ attempts. Each configuration had 500 attempts, including 360 expected non-admissions and 140 expected admissions. Both strict paths and a conventional structured broker recorded zero admissions against those 360 authored non-admission labels. This supports the need for context and effect predicates in those cases, not exclusive OSuite capability or independently labeled accuracy.

The study used real child processes, HTTP and SQLite, with a separate authenticated receiver journal. Its effect was a synthetic SaaS-style journal append, not a customer refund. Prepared admission had a gate p95 of 0.741917 ms versus 1.078042 ms fresh, but elapsed p95 was 54.291958 ms versus 52.848292 ms. End-to-end speedup was not established. Preparation costs, known effect counts and incomplete obligation matching prevent a general performance comparison. These figures are historical engineering evidence, not measurements of the selective rule in Table~\ref{tab:model}.

BAF exact-context/tuple checking and its attenuation checker remain distinct from selective evidence reuse. AREG's occurrence graph and supplied-world reachability bounds remain distinct from prospective effect inference. ZeroGate carries prepared authorization but must retain current authority and durable consumption. The historical study is not added to Matoi's process-campaign denominator.

\section{Portable and Integrated Artifact Accounting}
The later finite-contract manuscript and its supplement provide the complete
portable contrasts, semantic bounds and all seventeen integrated schedules.
The portable formal archive has manifest SHA-256
\path{19f5aa924e95b47563e61685a24c382d19b22206f5492f41bf891d81da4899b1}
and pre-data freeze
\path{103d09b68e9567c33f29f2775936047e6bbe34e60899c07809a8d5cf94e59cc8}.
The integrated archive has manifest
\path{28d1502cd87af752b618075728d7f424b65b8482b0b2a0ea5d8e1d37b9d6b29e}
and freeze
\path{fabcfce9d8c4346493a06cf65894bb9be208dc3de7ca1f447e76bea223ea82c3}.
The kernel/source-harness checkpoints are respectively
\texttt{18c472a8} and \texttt{a76f3b1c}. The earlier five-strategy repair
archive remains separate, with manifest
\path{e4ef9f6d6d5c3aa29d16345f309687eccbcf86c1c9aecd261534abcbaace62f1}.

The archives retain captured sources and source/input identities, raw output,
complete matrices and replay instructions. Integrated replay checks public-key
signatures, exact bodies/attempts, journal transitions and reopened receiver
and runtime databases without resending a mutation. It is a same-project
verification path, not an outside attestation of capture or timing. Local
archive names and hashes are not public download endpoints. This candidate
does not yet provide a public raw-artifact release. A local conformance bundle
now contains 25 unchanged source files with scoped MIT/CC BY 4.0 notices; two
separate clean exports each pass 128 JS and 70 Python tests. It has not been
uploaded and does not reproduce the paper's timing and execution claims.

\section*{Acknowledgments and Research Integrity}
OpenAI Codex assisted substantively with drafting and revision of the abstract,
all prose sections and appendices of this preprint and its ancillary
manuscripts. Affected material includes research framing and related work,
models and proof arguments, implementation and evaluation descriptions,
figure/table presentation, limitations, conclusions and disclosures.
Assistance also covered research scoping, authored tests and expectations,
experiment/verification code, orchestration, analysis and internal review.
The motivation was rapid implementation and cross-checking of the bounded
prototype, not replacement of measured execution. This was not limited to
copyediting. Reported numbers come from retained program outputs; AI-assisted
authored expectations and same-project separate-language checking are not
independent human labels or external validation. No unverified model version
is asserted; the development assistant's resource consumption was not
independently measured. The named author remains responsible for verification,
attribution, final claims and personal submission review; that final review is
not attested here. These experiments did not consume the
reserved skill-boundary collaboration holdouts and do not imply coauthorship,
review or institutional endorsement by prospective collaborators. No customer
data or third-party attack infrastructure was used. The Toronto affiliation
does not attest the geographical location of every experiment host.

\begingroup\small\raggedright

\endgroup
\end{document}